\documentclass[a4paper,UKenglish,cleveref, autoref, thm-restate]{lipics-v2021}

\usepackage[algo2e,ruled,lined,linesnumbered,noend]{algorithm2e}

\usepackage{caption}
\usepackage{subcaption}

\usepackage{xcolor}
\def\draft{1}  
\if\draft 1

\newcommand{\guy}[1]{{\color{red} {\bf Guy:} #1}}
\else

\newcommand{\guy}[1]{}
\fi

\newcommand{\LIB}{log-interleave bound}
\newcommand{\LIBcaps}{Log-Interleave Bound}

\newcommand{\APMcaps}{Adaptive Parallel Mergesort}
\newcommand{\IB}{interleave bound} 
\newcommand{\IBcaps}{Interleave Bound}

\newcommand{\ib}{\text{IB}}
\newcommand{\lib}{\text{LIB}}

\newcommand{\MoP}{measure of disorder}
\newcommand{\MoPs}{measures of disorder}

\title{The Geometry of Tree-Based Sorting} 

\author{Guy E. Blelloch}{Carnegie Mellon University, USA}{guyb@cs.cmu.edu}{}{}

\author{Magdalen Dobson}{Carnegie Mellon University, USA}{mrdobson@cs.cmu.edu}{}{}

\authorrunning{G. Blelloch and M. Dobson} 

\Copyright{Guy E. Blelloch and Magdalen Dobson} 

\ccsdesc[500]{Theory of computation~Sorting and searching} 

\keywords{binary search trees, sorting, dynamic optimality, parallelism} 

\category{} 

\relatedversion{} 

\acknowledgements{We thank Kanat Tangwongsan for helpful discussions. We thank the anonymous reviewers for their useful comments. This work was supported by the National Science Foundation grants CCF-1901381, CCF-1910030, CCF-1919223, DGE1745016, and DGE2140739.}

\nolinenumbers 

\ArticleNo{22}

\begin{document}

\maketitle

\begin{abstract}
We study the connections between sorting and the binary search tree (BST) model, with an aim towards showing that the fields are connected more deeply than is currently appreciated. While any BST can be used to sort by inserting the keys one-by-one, this is a very limited relationship and importantly says nothing about parallel sorting.   We show what we believe to be the first formal relationship between the BST model and sorting. Namely, we show that a large class of sorting algorithms, which includes mergesort, quicksort, insertion sort, and almost every instance-optimal sorting algorithm, are equivalent in cost to offline BST algorithms. Our main theoretical tool is the geometric interpretation of the BST model introduced by Demaine et al.~\cite{demaine2009bst}, which finds an equivalence between searches on a BST and point sets in the plane satisfying a certain property. To give an example of the utility of our approach, we introduce the \LIB{}, a measure of the information-theoretic complexity of a permutation $\pi$, which is within a $\lg \lg n$ multiplicative factor of a known lower bound in the BST model; we also devise a parallel sorting algorithm with polylogarithmic span that sorts a permutation $\pi$ using comparisons proportional to its \LIB. Our aforementioned result on sorting and offline BST algorithms can be used to show existence of an offline BST algorithm whose cost is within a constant factor of the \LIB{} of any permutation $\pi$. 
\end{abstract}

\section{Introduction}
\label{sec: intro}

Comparison-based sorting and searching on a BST are among the most elementary, important, and well-studied algorithmic topics in all of theoretical computer science. It has long been observed that they are closely related: both enjoy better performance on sequences that are information-theoretically simpler, such as reversals of sorted lists, sequences with long runs of consecutive keys, or sequences composed from simple shuffles of sorted lists. Indeed, their respective searches for instance optimality have yielded the independent discovery of almost identical results~\cite{petersson1995adaptive}. Despite the extensive number of similar results throughout the literature, there is comparatively very little known about the \textit{formal} relationship between sorting and the binary search tree (BST) model. In this paper, we present what we believe to be the first formal relation between the sorting cost model and the BST model. Our result shows that a large class of sorting algorithms, which includes mergesort, quicksort, insertion sort, and most adaptive sorting algorithms, are equivalent in cost to offline algorithms in the BST model. As this class is large and contains many well-studied algorithms, we find it interesting that we are able to show any kind of new and nontrivial theoretical result based on this characterization.

\textbf{Binary Search Trees.} The binary search tree is a fundamental data structure that stores an ordered universe of keys in a dynamic tree. A search for a key begins with a pointer to the root of the tree and at each step performs one of two unit-cost actions: move the pointer to a parent or child node, or perform a rotation. Rotations are key to understanding the power of the model, because they allow frequently queried elements to be kept close to the root of the tree as well as exploiting other kinds of order in the sequence of queried keys. The power of rotations is related to the concept of \textit{instance optimality}, a general term which refers to the fact that an algorithm can enjoy improvements on its worst-case complexity on well-defined sets of ``easy'' inputs. Many BST algorithms perform better than their worst-case complexity (i.e. $\log n$ per operation) on various kinds of input~\cite{wang2006competitive, demaine2007tango, demaine2009bst, badoiu2007unified}. Beyond these specific improvements, the hope for a more general kind of instance optimality is crisply expressed by the \textit{dynamic optimality conjecture} of Sleator and Tarjan~\cite{sleator1985self}, which states that there exists a binary search tree whose performance on any online sequence of searches is constant factor competitive with the best offline algorithm for that sequence.
The dynamic optimality conjecture remains open.
Another equally important open question is whether there is an offline efficient algorithm for calculating
even an approximately optimal number of rotations for a given input sequence.
These problems have been the subject of extensive work both in the past~\cite{wilber1989bounds, cole2000dyn_pt1, cole2000dyn_pt2, demaine2007tango, derryberry2009thesis, bose2014lazy} and at the present moment~\cite{lecomte2020wilber, chalermsook2018multi, bose2020competitive, kozma2018smooth, iacono2016weighted}.
 
\textbf{Sorting.} Similarly to the BST model, where rotations are used to achieve instance optimality for particular classes of inputs, in comparison-based sorting an \textit{adaptive sorting algorithm} performs fewer comparisons when the input is ``closer'' to sorted by some measure. In this field a \textit{\MoP{}} for a list $L$ is paired with an algorithm which is \textit{optimal} for this measure. Here, optimal roughly means that sorting $L$ only requires the number of comparisons needed to distinguish it from all other lists which are more presorted than $L$~\cite{petersson1995adaptive}. 
An accompanying notion is that a \MoP{} may be \textit{superior} (inferior) to another measure---that is, always requires fewer comparisons for any given permutation. Mannila first formalized these ideas~\cite{mannila1985presorted}. After this, many researchers devised new \MoPs{} and corresponding optimal algorithms~\cite{cook1980sort, castro1989sort, katajainen1989insertion, levcopoulos1990shuffled, levcopoulos1991splitsort, levcopoulos1993heapsort, moffat1990hist, petersson1995adaptive}; furthermore, there was also interest in work-optimal parallel versions of optimal sorting algorithms~\cite{Carlsson1991adaptive, levcopoulos1996inversions, chen1992improved}. It remains an open problem whether there exists a measure which is provably superior to any possible \MoP.

\textbf{Arborally Satisfied Point Sets and Sorting.} One of our most important tools in connecting BSTs and sorting is the geometric interpretation of BSTs~\cite{demaine2009bst, DSW05}. In this interpretation, an access sequence of $n$ keys is represented as an $n \times n$ grid with time order (input order) on one axis (here the $x$ axis) and key order (output order) on the other axis. Points are added to the grid to account for all keys that must be visited when searching or inserting the keys one at a time from left to right. Demaine et al.~\cite{demaine2009bst}, and Derryberry,  Sleator, and Wang~\cite{DSW05} show that for any BST algorithm, the accesses plotted in the plane must satisfy the property that for every pair of points $p, q$ (both original and added points), there is a monotonic path (e.g. up and right) from $p$ to $q$ consisting of horizontal and vertical segments with a point at each corner\footnote{The papers have their own preferred definition, but the definitions are all equivalent.}.  Demaine et al. refer to such a set of points as being \emph{arborally satisfied}, and show that any such set of size $m$ implies the sequence of $n$ keys can be searched or inserted in cost $O(m)$ in the BST model. 

We observe that the geometric approach is also useful for sorting, since unlike the BST model, it does not directly enforce an order of insertion. As some evidence of the utility of the geometric approach, consider the following two algorithms that can be used to arborally satisfy a set of accesses. The first algorithm starts by choosing a random point, then adding accesses to that point across its entire row of the point set. Then, it recurses above and below the row, and in each partition it picks a random point and adds new points to all locations in its row
for which there is a key in the partition.  It is not hard to verify the points added in this way are arborally satisfied:  any point $p$ can
get to a point $q$ by going up (or down) to the row that separated them, then across to column of $q$ and up (down) to $q$.
Second, consider another algorithm that adds points across the middle column, and then for the left and right, add points along their middle columns for all points in those halves. Recursing to the base case again gives an arborally satisfied set. See Figure~\ref{fig: arboralSatisfaction} for an example of both of these algorithms. The attentive reader may have noticed that the accesses added in the first algorithm correspond to the comparisons made by the quicksort algorithm, and the accesses added in the second algorithm correspond to comparisons made by the mergesort algorithm.  We will extend these ideas to more interesting algorithms in this paper.

\begin{figure}
	\centering
	\begin{subfigure}[t]{.3\textwidth}
		\centering
		\includegraphics[width=\textwidth]{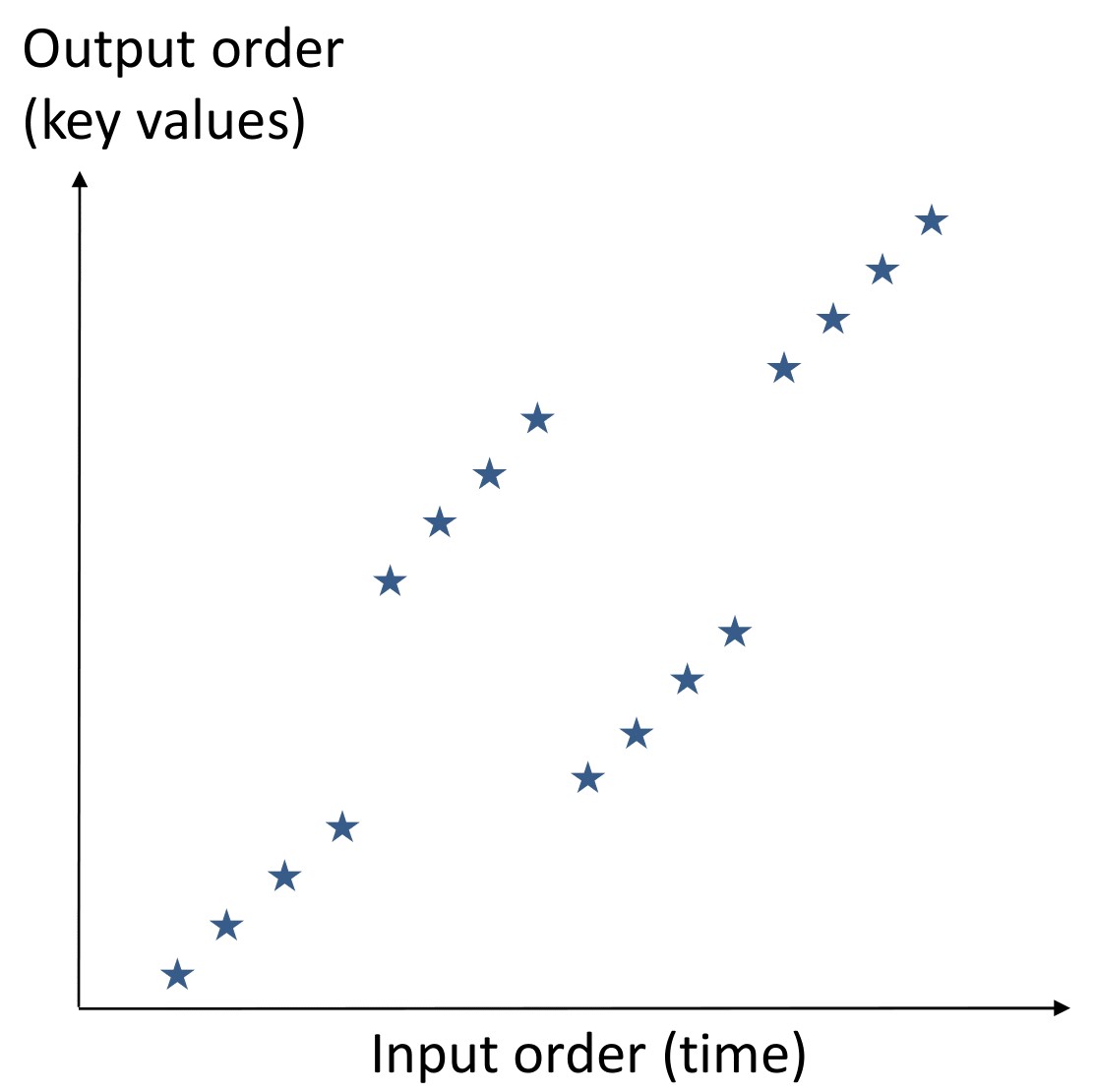}
		\caption{The input sequence}
		\label{fig: input}
	\end{subfigure}
	\hfill
	\begin{subfigure}[t]{.3\textwidth}
		\centering
		\includegraphics[width=\textwidth]{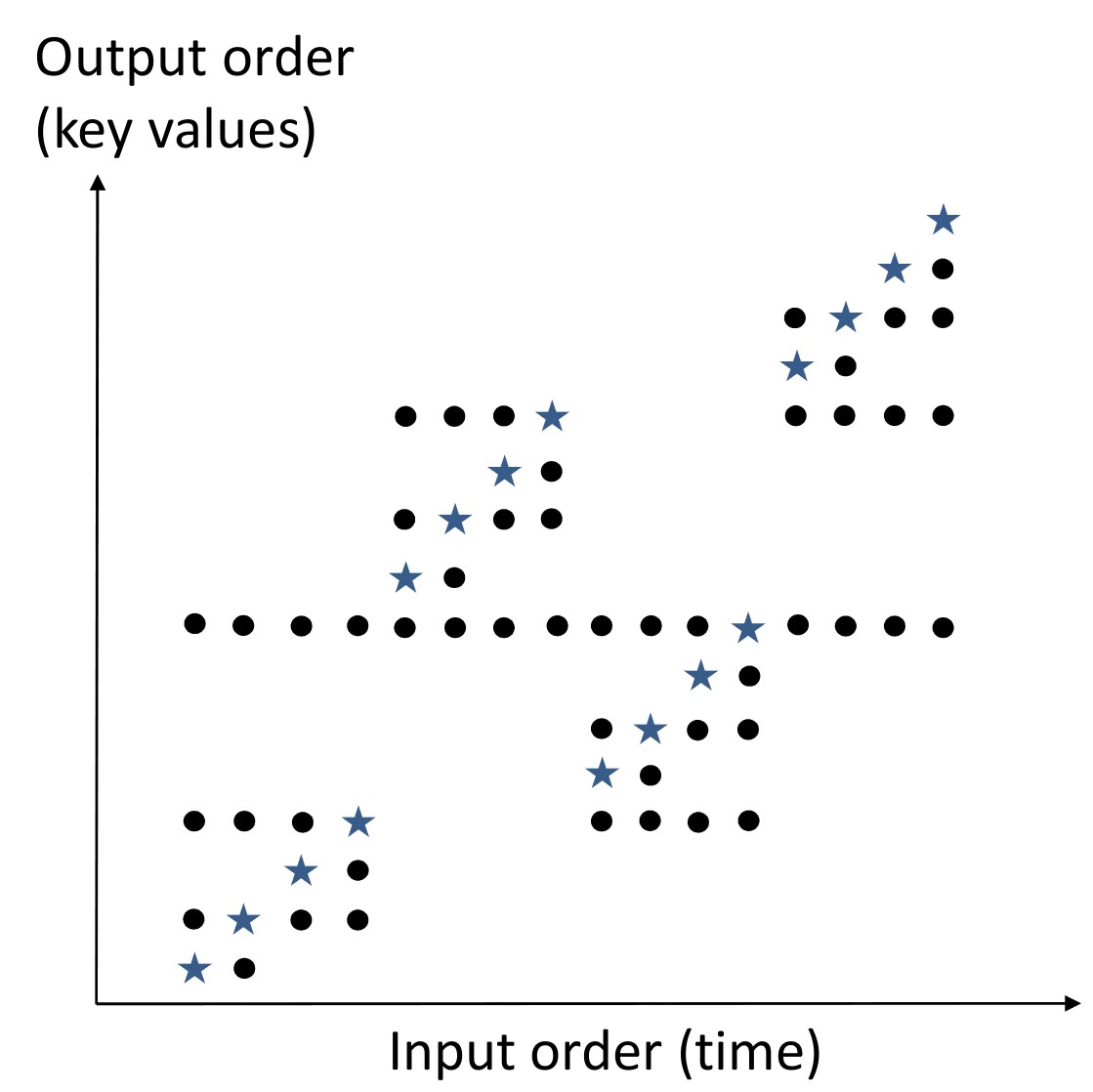}
		\caption{Sequence arborally satisfied using quicksort}
		\label{fig: quicksort}
	\end{subfigure}
	\hfill
	\begin{subfigure}[t]{.3\textwidth}
		\centering
		\includegraphics[width=\textwidth]{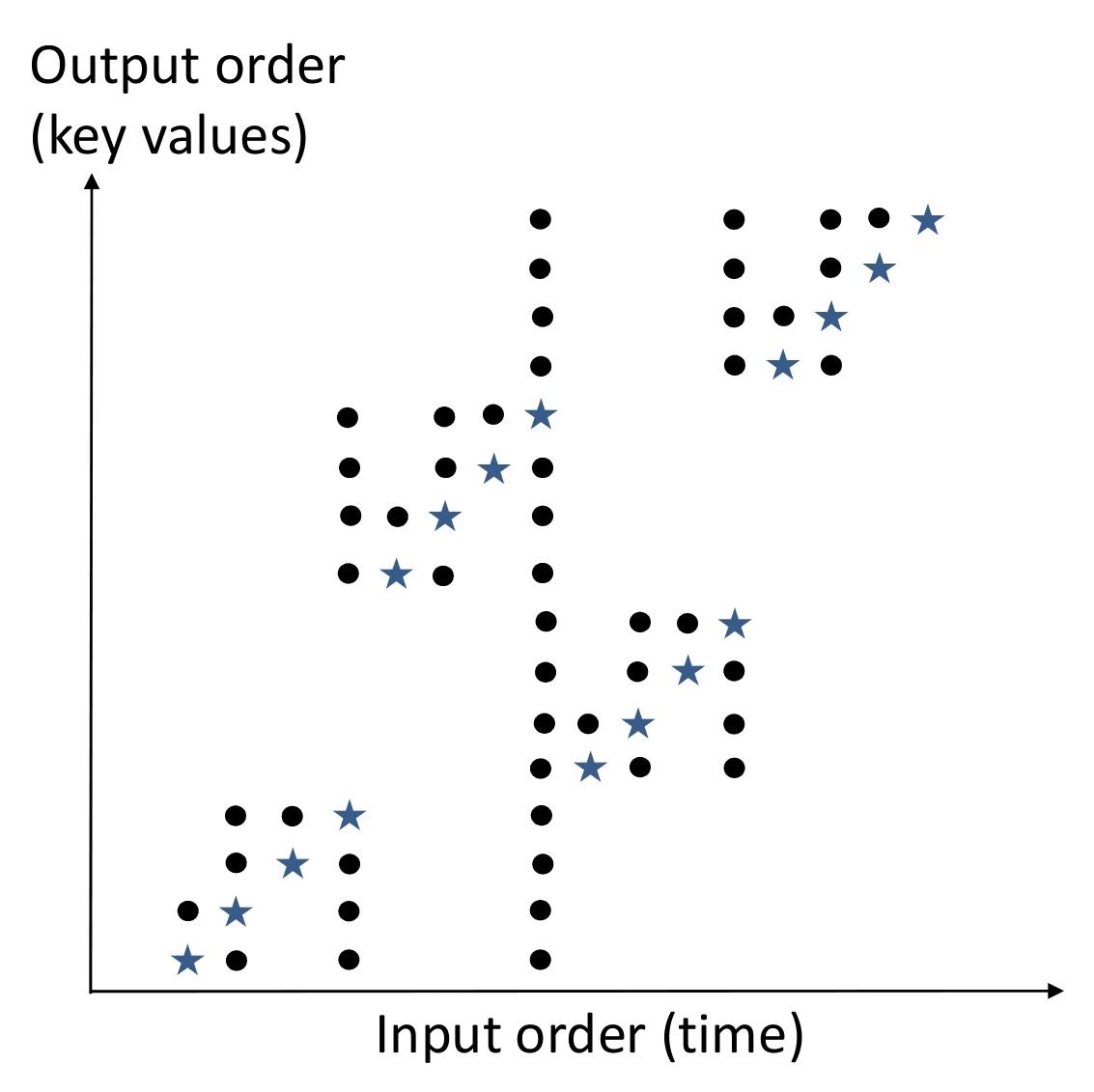}
		\caption{Sequence arborally satisfied using mergesort}
		\label{fig: mergesort}
	\end{subfigure}
	\caption{On the left, the input sequence plotted in the plane. In the middle, the input sequence arborally satisfied using accesses corresponding to quicksort. On the right, the input sequence arborally satisfied using accesses corresponding to mergesort.}
	\label{fig: arboralSatisfaction}
\end{figure}

We note that in addition to being of significant theoretical interest, taking advantage of locality in key sequences is widely practical for both search trees and sorting.   Sleator and Tarjan won the ACM Kannelakis Theory and Practice award for their work on splay trees and its applications for reducing key search time in several widely used applications.   Adaptive sorting algorithms are widely adopted in practice, including \textit{timsort}, which is implemented as built-in libraries for Python, Java, Swift, and Rust, among other languages~\cite{auger2018timsort}.

\subsection{Our Results.}

In this paper we present specific results relating sorting and BSTs using arborally satisfied sets. Our first result is an explicit relation between the cost models of a broad set of sorting algorithms and BST algorithms. The specific set of algorithms, which we refer to as \textit{tree-based} sorting algorithms and which are formally defined in Section~\ref{sec: sorting}, are divided into two classes: \textit{BST mergesorts} are sorting algorithms based on recursive merges of the input, where the keys being merged are stored in binary search trees; \textit{BT partition sorts} are sorting algorithms based on recursive partitions of the input based on key ordering, where the keys are stored in a binary tree. In Section~\ref{sec: sorting}, we show that these two types of algorithms are ``dual'' to each other in the following way: a BST mergesort $\mathcal{A}$ sorting permutation $\pi$ with cost $\mathcal{A}(\pi)$ implies the existence of binary tree (BT) partition sort $\mathcal{B}$ sorting $\pi^{-1}$ with cost $\mathcal{A}(\pi)$, and vice versa.

The category of tree-based sorting algorithms is large. Both quicksort and mergesort on lists fall into the category of tree-based sorting algorithms, as lists are simply a special case of trees. Sorting by insertion into a BST is also a BST mergesort, since the merges can be carried out in any order. McIlroy's adaptive sorting algorithms---namely, insertion sort with exponential search and mergesort with exponential search---are BST mergesorts~\cite{mcilroy1993sorting}. These algorithms were influential in the development of timsort~\cite{auger2018timsort}, a mergesort algorithm that breaks the input into runs of increasing or decreasing keys and merges them based on certain ordering criteria; since it is also a mergesort on lists, timsort is a BST merge. In their capstone paper on adaptive sorting, Petersson and Moffat cover the three most powerful known adaptive sorting algorithms---local insertion sort, historical insertion sort, and regional insertion sort. Local insertion sort is a BST mergesort as it inserts into a BST with a single additional pointer, which can be converted to the BST model with constant overhead~\cite{chalermsook2018multi}. Historical and regional insertion sort are not BST mergesorts as presented by the authors, but independently discovered data structures would yield BST mergesorts with the same bounds~\cite{sleator1985self, derryberry2009thesis}.

Theorem~\ref{thm: dualoffline} shows that for any tree-based sorting algorithm that sorts access sequence $\pi$ using $\mathcal{A}(\pi)$ accesses, there exists an offline algorithm in the BST model which searches for each key in $\pi$ using $O(\mathcal{A}(\pi))$ accesses. The proof of Theorem~\ref{thm: dualoffline} is subtle and nontrivial and requires several new insights relating to the geometric interpretation of the BST model. In the following statement, $\text{OPT}_{\text{BST}}(\pi)$ refers to the cost of the best offline algorithm.

\begin{theorem}\label{thm: dualoffline}
Let $\mathcal{A}$ be a tree-based sorting algorithm which sorts permutation $\pi$ using $\mathcal{A}(\pi)$ accesses. Then  $\text{OPT}_{\text{BST}}(\pi) \in O(\mathcal{A}(\pi))$.
\end{theorem}

As some evidence for the utility of our approach, we introduce the \LIB{}, a measure of the information-theoretic complexity of a permutation $\pi$. The \LIB{} is an upper bound on the number of bits needed to encode $\pi$; it can also be understood from an algorithmic perspective as a mergesort with a more efficient merge step. Our main results on the \LIB{} illustrate the connections between sorting and the BST model. In the statements of the following results, we use the notation $\lib(\pi)$ to refer to the \LIB{} of a permutation $\pi$. This will be defined formally in Section~\ref{sec: log-interleave}.

The first result is a proof that the \LIB{} is within a $\lg \lg n$ multiplicative factor of the optimal offline BST algorithm on any permutation. Somewhat similarly to Demaine et al.'s proof of the closeness to optimality of tango trees~\cite{demaine2007tango}, our proof shows closeness to optimality by comparing the \LIB{} with Wilber's \IB~\cite{wilber1989bounds}, a lower bound in the BST model. 

\begin{restatable*}{theorem}{thmlib}\label{thm: lib}
	For any permutation $\pi$, $\ib(\pi) \leq \lib(\pi) \in O(\lg \lg n \; \ib(\pi))$.
\end{restatable*} 

Next, we show that there is a work optimal \textit{parallel} sorting algorithm related to the \LIB. The next result is a parallel mergesort featuring a merge step which combines recent work on parallel split and join of BSTs~\cite{blelloch2016just} with a BST from~\cite{brown1979fast} and an analysis which shows that with this new merge step, the mergesort sorts a sequence $\pi$ in $O(\lib(\pi))$ work. While the span of this algorithm is greater than the span of a typical parallel sort and indeed may be open to improvement, all existing parallel sorting algorithms present guarantees only for very weak \MoPs~\cite{chen1992improved, levcopoulos1996inversions}. 

\begin{restatable*}{theorem}{thmsort}\label{thm: mergesort}
There exists a parallel mergesort which for any permutation $\pi$  performs $O(\lib(\pi))$ work with polylogarithmic span.
\end{restatable*}

Finally, a corollary of Theorem~\ref{thm: offlinealg} shows that there is an offline BST algorithm that incurs cost $O(\lib(\pi))$, and thus that the \LIB{} is an upper bound in the BST model.

\begin{restatable*}{corollary}{corlib}\label{cor: offlinelib}
There exists an offline BST algorithm $\mathcal{A}$ such that $\mathcal{A}(\pi) = O(\lib(\pi))$.
\end{restatable*}

\textbf{Model of Computation.} Our results for the parallel algorithms are given for the binary-fork-join model~\cite{BlellochF0020}.  
In this model a process can fork two child processes, which work in parallel and when both complete, the parent process continues.   Costs are measured in terms of the work (total number of instructions across all processes) and span (longest dependence path among processes).  Any algorithm in the binary forking model with $W$ work and $S$ span can be implemented on a CRCW PRAM with $P$ processors in $O(W/P + S)$ time with high probability~\cite{ABP01,blumofe1999scheduling}, so the results here are also valid on the PRAM, maintaining work efficiency.

\subsection{Related Work.}
\label{sec:related}

\textbf{Upper and Lower Bounds in the BST Model.} The pursuit of dynamic optimality led to a string of work in both upper and lower bounds on the cost of a sequence of searches on a BST. Three important upper bounds in the literature are the \textit{dynamic finger bound}~\cite{sleator1985self, cole2000dyn_pt1, cole2000dyn_pt2, chalermsook2018multi, bose2014lazy, iacono2016weighted}, the \textit{working set bound}~\cite{sleator1985self}, and the \textit{unified bound}~\cite{badoiu2007unified, derryberry2009thesis}, which respectively state that accessing an element is fast if its key is close to the key of the previous search, if its key has been searched recently, and a combination of the two. There has also been significant work in lower bounding the cost of an access sequence in the BST model. Two such lower bounds, the \textit{interleave bound} and the\textit{ funnel bound}, were introduced by Wilber in~\cite{wilber1989bounds}; a recent work by Lecomte and Weinstein~\cite{lecomte2020wilber} affirmatively settled the 30-year open question of whether the funnel bound was tighter than the interleave bound, proving a $\lg \lg n$ multiplicative separation in some cases. Another lower bound, the \textit{rectangle bound}, was introduced by Demaine et al. in~\cite{demaine2009bst}.

\textbf{Progress on Dynamic Optimality.} The BST which comes closest to dynamic optimality is the \textit{tango tree} of Demaine et al.~\cite{demaine2007tango}, which has a competitive ratio of $O(\lg \lg n)$ with respect to the best offline algorithm. Wilber's interleave bound was vital in the analysis of the competitive ratio, since the authors showed that on any access sequence $x$, the tango tree uses $O(\lg \lg n \; \ib(x))$ accesses, where $\ib(x)$ represents the interleave bound of the sequence. In~\cite{wang2006competitive}, Wang et al. introduce the multi-splay tree, a BST which achieves $O(\lg \lg n)$ optimality with better worst-case guarantees than the tango tree. Prominent candidates for a dynamically optimal algorithm include the splay tree, which was presented by Sleator and Tarjan at the same time as the dynamic optimality conjecture~\cite{sleator1985self}, and the Greedy algorithm presented in Demaine et al.'s geometric interpretation of the BST model~\cite{demaine2009bst}.

\textbf{Other Data Structures and the BST model.} The \textit{min-heap}, which stores a set of keys and supports inserting arbitrary elements and extracting and deleting the minimum element. Recently, Kozma and Saranurak show an explicit relation between the BST model and the heap cost model~\cite{kozma2018smooth}, as well as proposing an analogue of the dynamic optimality conjecture for heaps. Specifically, they show that for every heapsort algorithm (that is, an algorithm which sorts a permutation $\pi$ with cost $\mathcal{A}(\pi)$ by inserting its keys into a heap and repeatedly extracting the minimum element) corresponds to an insertion sort into a BST algorithm which incurs cost $\mathcal{A}(\pi)$ on the inverse permutation $\pi^{-1}$. Their insight came from relating the rotation operation in a BST to the link operation in a heap, which allowed them to relate the corresponding cost models.

\textbf{Adaptive Sorting in Parallel.} During the period of interest in the adaptive sorting model, researchers were also interested in work-optimal parallel sorting algorithms with polylogarithmic span. Unsurprisingly, such algorithms exist for practical measures Runs and Inv~\cite{Carlsson1991adaptive, chen1992improved}; one also exists for Osc, a generalization of Inv~\cite{levcopoulos1996inversions} that is still theoretically weak. To our knowledge there are no results on parallel sorting algorithms which are optimal with respect to any stronger measures.

\section{Preliminaries}\label{sec: prelims}
\textbf{Terminology.} Throughout this paper, we will use the terms \textbf{list}, \textbf{permutation}, and \textbf{access sequence} interchangeably to refer to some ordering of the keys $1, 2, \ldots, n$. The term access sequence is used in the literature on BSTs to denote a sequence of queries to a BST; unless otherwise stated, an access sequence is presumed not to contain repeated keys.

\textbf{The Binary Search Tree Model.} 
A binary tree (BT) is either a \emph{leaf} or a \emph{node} consisting of a left binary tree, a key and a right binary tree.  A binary seach tree (BST) is a binary tree where the keys have a total order, and for each node in the tree all keys in its left subtree are less than its key, and all keys in its right tree are greater.

The following definition of the BST model is drawn from~\cite{sleator1985self, wilber1989bounds, demaine2007tango}. 
The model assumes an initial BST with keys $[1,2,\ldots,n]$ and an access sequence
$[x_1, x_2, \ldots, x_m]$ of searches, where $x_i \in \{1, 2,\ldots,n\}$. 
Each search starts at the root and at each node it visits, it may perform one of the following actions: (a) move to the right child, left child, or parent, or (b) perform a rotation of the node and its parent. Each of these actions has unit cost and the search must
visit its specified key. 
We refer to an algorithm that decides on what actions to perform for each search as a \textbf{BST algorithm}.  A BST algorithm may be offline---meaning it can see the entire sequence of queries ahead of time---or online, meaning that queries are revealed one at a time.

\textbf{Wilber's \IBcaps.} Wilber's interleave bound is a lower bound on the cost of accessing any sequence in the BST model. Given an access sequence $\pi$ consisting of the keys $x_1, x_2, \ldots, x_n$, fix a static binary tree $P$ (meaning it will never be rotated) with the keys of $\pi$ at the leaves in the order they appear in $\pi$. Calculate the \IB{} of $\pi$ as follows: query the keys in $\pi$ in sorted order. For each vertex $v_j$, label each element $i$ of the sequence with R or L, depending on whether accessing $i$ in $P$ goes through the right or the left subtree of $v_j$, respectively (if $i$ is in neither subtree, give it no label). The \IB{} of $v_j$, denoted $\ib(v_j)$, is the number of switches between $R$ and $L$ in the labels if the keys are queried in sorted order. The \IB{} of the entire access sequence $\pi$ is calculated by summing over the \IB s of each vertex, so $\ib(\pi) = \sum_{v \in P} \ib(v_i)$. As the lower bound holds for an arbitrary tree $P$, the \IB{} of the sequence is usually understood to refer to the maximum over all static trees. See Figure~\ref{fig: IB} for an example calculation.

\textbf{Arborally Satisfied Sets.} In~\cite{derryberry2009thesis}, Derryberry et al. formalize a connection between binary search trees and points in the plane satisfying a certain property. An access sequence can be plotted in the plane where one axis represents key values and the other axis represents the ordering of the search sequence (that is, time). In the context of sorting, these axes can also be referred to as input order and output order. In this work, we use the horizontal axis for time and the vertical access for keyspace. See Figure~\ref{fig: arboralSatisfaction} for an example of an arborally satisfied set.

In addition to plotting the search sequence on the plane, one can also plot the key values of the nodes which a BST algorithm accesses (for search or rotations) while searching for a node. When searching for an element $x_i$ which is inserted at time $i$, the values of the nodes in the search path are plotted on the same vertical. Demaine et al.~\cite{demaine2009bst} proved that such a plot satisfies the following property:

\begin{definition}\label{def: arborallysatisfiedset}
Given a set $P$ of points in the plane, $P$ is \textbf{arborally satisfied} if for every two points $x,y \in P$ that are not on the same vertical or horizontal, the rectangle defined by $x$ and $y$ contains at least one point in addition to $x$ and $y$.
\end{definition}

As mentioned in Section~\ref{sec: intro}, a useful equivalent definition of arboral satisfaction is that there must be a monotonic path (e.g. consisting only of moves up and to the right) between $x$ and $y$ with a point at every corner. Note that any valid search on a BST will only touch nodes in a subtree $\tau_i$ of tree $T$, where $\tau_i$ includes the root of $T$. We sometimes refer to such a subtree as a \textbf{top tree} of $T$.

Demaine et al. show that a BST can be used to arborally satisfy an access sequence plotted in the plane. However, one can also use an algorithm that directly places points in the plane rather than using a BST.

\begin{definition}\label{def: arboralalg}
Given a set of points in the plane corresponding to an access sequence $\pi$, an \textbf{offline arboral satisfaction algorithm} adds points to the plane to make an arborally satisfied set. An \textbf{online arboral satisfaction algorithm} also adds points in the plane to form an arborally satisfied set, but accesses are revealed one by one in input order and the algorithm must produce an arborally satisfied set at each time.
\end{definition}

Demaine et al. show that a BST algorithm is equivalent to an arboral satisfaction algorithm, but they also show a more surprising result: an arboral satisfaction algorithm is equivalent to a BST algorithm. Specifically, they show that an offline (online) arboral satisfaction algorithm requiring $f(\pi)$ accesses to arborally satisfy a search sequence $\pi$ can be transformed to an offline (online) BST algorithm requiring $O(f(\pi))$ accesses to search for the elements of a sequence $\pi$.

\section{Tree-based Sorting}\label{sec: sorting}

Towards the goal of unifying the BST model and sorting, we ask the following question: when can the costs of a sorting algorithm be related to the costs of a BST algorithm? Clearly not every comparison-based sorting algorithm should be relatable to the BST model: as mentioned in Section~\ref{sec: intro}, for example, an algorithm that guesses and checks could take $O(n)$ steps for an arbitrary permutation.  Our investigation is therefore limited to sorting over binary trees, and in particular it considers a class of mergesort and partition sort (related to quicksort) algorithms on binary trees.

\subsection{Mergesort.}\label{sec: mergesorting}
We first consider mergesorts in which the sequences to be merged are represented as BSTs.
``Mergesort'' is interpreted here as merging based on any split of the input sequence, not just splits into equally-sized parts; thus insertion sort using a sequence of insertions into a BST is a special case of mergesort.  The cost of these algorithms is measured in terms of the number of accesses to the tree required during the mergesort, which will always be at least as great as the number of comparisons. We capture the idea of a BST mergesort more formally before giving the theorem statement.

A merge will interleave contiguous subsequences from its two inputs.  We refer to each of these subsequences as \emph{blocks} and we refer to the ends of each block as \emph{block boundaries}.  The block boundaries need to be accessed to even verify that the merge is correct.  The following defines a merge that examines some top part of two trees to generate its output.

\begin{definition}\label{def: bstmerge}
  A \textbf{BST merge} takes two BSTs $T_A$ and $T_B$, and for some top trees $\tau_a$ of $T_A$ and $\tau_b$ of $T_B$, returns a BST $T$ such that for some top tree $\tau$ of $T$,  $\tau =\tau_a \cup \tau_b$, the subtrees of $\tau$ correspond to unchanged subtrees of $\tau_a$ and $\tau_b$, and $\tau$ contains the block boundaries.  The \emph{number of accesses} used by the merge is $|\tau|$. 
\end{definition}

\begin{definition}\label{def: bstmergesort} A \textbf{BST mergesort} recursively splits the input sequence into two parts each of size at least 1, sorts each part with a BST mergesort, and executes a BST merge on the results.  A mergesort on an input of size 1 returns its input.  The \emph{number of accesses} used by the mergesort is the sum of accesses across all merges.  \end{definition}

This leads to the main theorem.
\begin{theorem}\label{thm: offlinealg}
	Let $\mathcal{A}$ be a BST mergesort algorithm which sorts permutation $\pi$ using $\mathcal{A}(\pi)$ accesses. Then  $\text{OPT}_{\text{BST}}(\pi) \in O(\mathcal{A}(\pi))$.
      \end{theorem}

Theorem~\ref{thm: offlinealg} is proved by showing that for every BST mergesort algorithm, there is an offline BST algorithm that incurs the same cost as the BST mergesort within a constant factor. Our proof relies on the geometric interpretation of the BST model---instead of directly transforming a BST mergesort into an offline BST algorithm we use arborally satisfied sets as an intermediary.  The key ingredient is a transformation of a BST merge algorithm to an offline arborally satisfied set algorithm, which is equivalent to an offline BST algorithm by Demaine et al.'s theorem~\cite{demaine2009bst}. The following graphic illustrates the chain of dependencies. 

\begin{center}
\includegraphics[scale = .6]{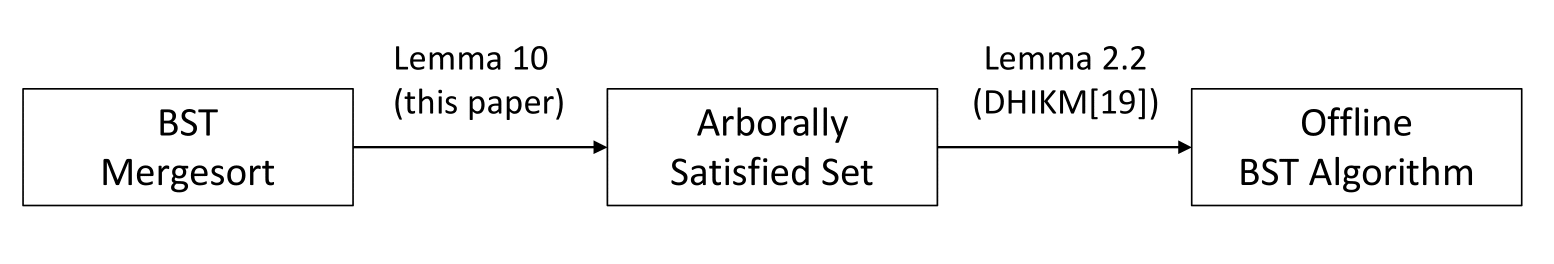}
\end{center}

\begin{algorithm2e}[H]
	\caption{ arboralMerge($A, B, \mathcal{M}$) \protect\\
		The arboral satisfaction algorithm; also illustrated in Figure~\ref{fig: arboralmerge}.}
	\label{algo: arboralmerge}\small
	\SetKwBlock{ParDo}{do in parallel}{end}
	\SetAlgoLined
	\KwIn{Two arborally satisfied sets $A, B$; BST merge algorithm $\mathcal{M}$.}
	\KwOut{An arborally satisfied set $C$ consisting of the concatenation of $A$ and $B$ as well as additional accesses needed to arborally satisfy the concatenation.}
	\lIf{$A == \emptyset$}{
		\KwRet{$B$} 
	} 
	\lIf{$B == \emptyset$}{
		\KwRet{$A$} 
	} 
	$C$ $\leftarrow$ concatenate $A$ and $B$ on the time axis \;
	$S$ $\leftarrow$ set of keys accessed when merging keys in  $A$ and $B$ using $\mathcal{M}$ \;
	access each key in $S$ in the first, middle
	(rightmost column of $A$), and end columns of $C$ \;
	\KwRet{$C$} \;	
	\vspace{0.5em}
\end{algorithm2e}

The main idea behind going from the mergesort to an arborally satisfied set is to transform a merge algorithm $\mathcal{M}$ into an algorithm that ``merges'' two arborally satisfied sets by concatenating them along the time axis (in this paper the $x$ axis) and resolving any unsatisfied rectangles between the two sets, thus producing an arborally satisfied set. This ``arboral'' mergesort algorithm would by definition be an offline arborally satisfied set algorithm. Ideally this arboral merge would use the same number of accesses as a corresponding BST merge $\mathcal{M}$. The key idea behind our algorithm is to use the keys accessed during the tree merge to arborally satisfy the sets on the two sides, as well as add points to make it easier to satisfy the condition on future merges.
In particular, the keys are added in three columns: the leftmost column of the left set, the rightmost column of the right set, and a middle column---we will use the rightmost of the left set, although the leftmost of the right set would also work.
Roughly, the accesses are placed down the middle to resolve unsatisfied rectangles, and they are placed down the leftmost and rightmost columns to restore invariants that are useful for future merges.
These accesses must by definition form top trees of the two trees being merged. Algorithm~\ref{algo: arboralmerge} shows the merging routine, which is used as a sub-step in an arborally satisfied set algorithm which recursively splits the input set to singletons and then uses Algorithm~\ref{algo: arboralmerge} to merge sub-parts of the input; see Figure~\ref{fig: arboralmerge} for an illustration of the merge step. Since we are adding a constant number of points per access in the mergesort, the total points added is proportional to the
cost of the mergesort, which in turn implies an offline BST algorithm with the same cost as the original BST mergesort.

The most difficult part of Theorem~\ref{thm: offlinealg} is proving correctness of the arboral mergesort---that is, that each execution of the arboral merge routine produces an arborally satisfied set. This will require some more background on arborally satisfied sets.

\begin{figure}
\centering
\begin{subfigure}[t]{.45\textwidth}
\centering
\includegraphics[width=\textwidth]{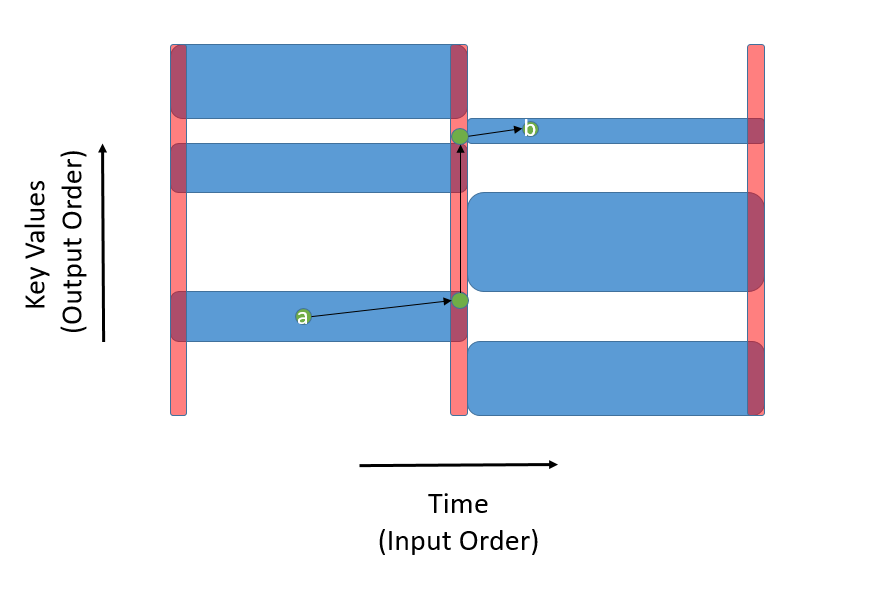}
\caption{An illustration of an arboral merge}
\label{fig: arboralmerge}
\end{subfigure}
\hfill
\begin{subfigure}[t]{.45\textwidth}
\centering
\includegraphics[width=\textwidth]{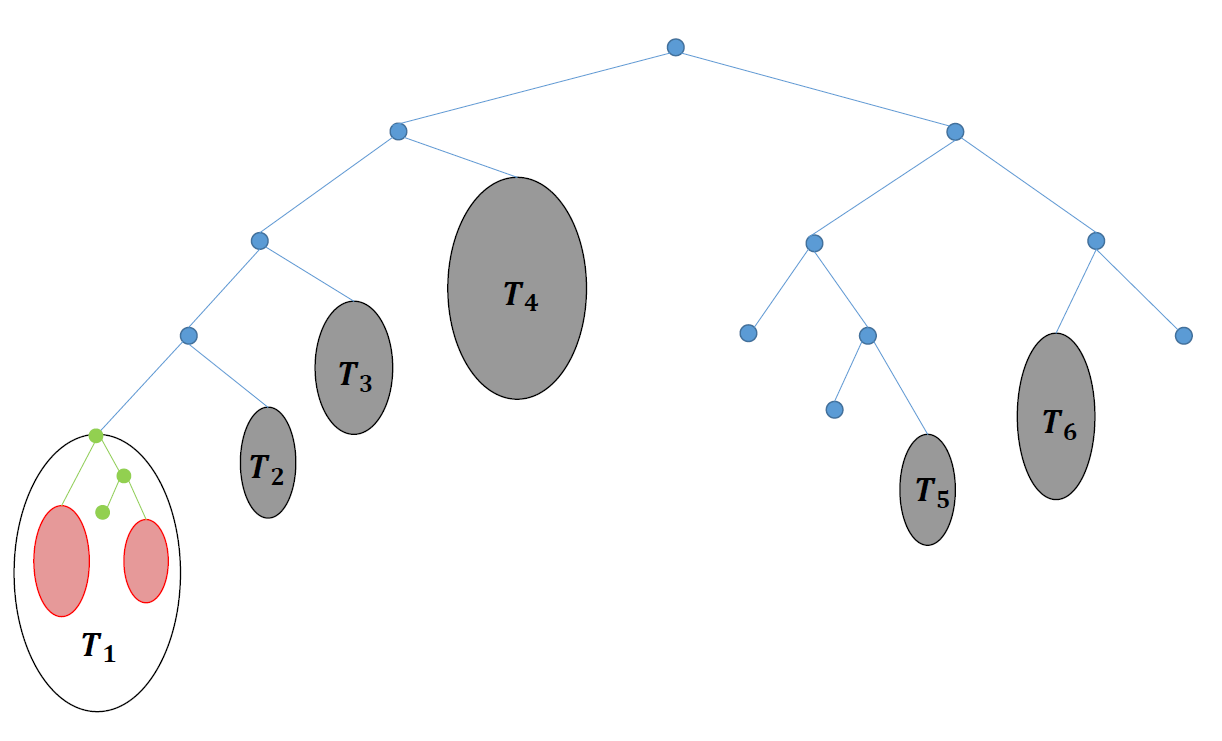}
\caption{A BST separated into a top tree and auxiliary trees}
\label{fig: treeArg}
\end{subfigure}
\caption{On the right, an illustration of the merging algorithm shown in Algorithm~\ref{algo: arboralmerge}. The blue squares represent members of the two sets being merged, while the red columns (referred to as $C_L, C_M, C_R$ in the proof of Theorem~\ref{thm: offlinealg}) illustrate the additional accesses necessary for the merge. A path drawn from $a \in A$ to $b \in B$ illustrates how the accesses in the middle column ensure the set is arborally satisfied. On the right, a BST with a top tree shown in blue and auxiliary trees $T_1$ through $T_6$, where the recursive structure is shown in $T_1$.}
\label{fig: dataStructures}
\end{figure}

\begin{definition}\label{def: treap}
A \textbf{treap} over a set of pairs $S$ is a BST over the first coordinate of each $s \in S$ and a min-heap over the second coordinate. Ties over the second coordinate are permitted and may be broken arbitrarily.
\end{definition}

A treap with ties on the priorities can also be expressed as a \textbf{multi-treap} (for multi-node treap), where ties are stored in a \textbf{multi-node} that may have more than two children. Child relations must obey an underlying BST structure on the nodes stored within a multi-node; hence a multi-node may have at most one more child than the number of keys in the multi-node. A key component of our proof is that we will define a multi-treap with respect to each side of an arborally satisfied set and then relate these to the BST the mergesort will generate.

\begin{definition}
Given an arborally satisfied set $A$, let the left (right) priority be the distance from the left (right) boundary of the first point in the row (closer has higher priority).   The \emph{left (right) multi-treap} of $A$ is the multi-treap defined by the (row, priority)  pairs.
\end{definition}

Since there can be many points in any given column of an arborally satisfied set, multi-nodes of the multi-treaps can have more than two children.\footnote{Demaine et al.~\cite{demaine2009bst} define a similar notion when proving that for any arborally satisfied set there is a BST execution with equivalent cost, but only with respect to one side, and only when sweeping column by column.}

\begin{definition}\label{def: congruence}
Given a BST $T$ and an arborally satisfied set $A$ with left (right) multi-treap $H_A$, $T$  is \emph{left (right) congruent with $A$} if there is some valid BST structure on $H_A$ (forming a binary subtree within each multi-node) such that the tree structure on $H_A$ is equal to $T$. $T$ is \emph{doubly} congruent with $A$ if it is both left and right congruent.
\end{definition}

Note that many BSTs can be left (equivalently right) congruent with the same arborally satisfied set $A$ due to equal priorities.  Also many arborally satisfied sets can be left (right) congruent with the same BST $T$.    
It may seem unlikely that a BST is doubly congruent with a arborally satisfied set, but in our construction we will maintain double congruence, and in particular we will
show that the point sets created by Algorithm~\ref{algo: arboralmerge} are doubly congruent to the corresponding tree.

The following observation will be useful for the proof of Theorem~\ref{thm: offlinealg}. It follows from a similar argument to Lemma 2.1 of Demaine et al.~\cite{demaine2009bst}.

\begin{observation}\label{observation}
Consider an arborally satisfied set $A$ that is double-congruent with a tree $T$.  Then for any top tree $\tau$ of $T$, if the keys in $\tau$ are accessed along either the left or right column of $A$, or one past the left or right column, the resulting set of points is arborally satisfied.
\end{observation}

\begin{proof}
	Begin with the case where the keys of $\tau$ are accessed one past the leftmost or rightmost column of $A$. Assume for the sake of contradiction that there exists an unsatisfied rectangle---that is, a rectangle with two points at its corners and no points contained within it---between access $a \in A$ and access $b \in \tau$. Consider the least common ancestor $c$ of $a$ and $b$, whose key value must be between those of $a$ and $b$. Since by our assumption, there are no keys accessed between the rectangle defined by $a$ and $b$, this contradicts the fact that $\tau$ is a continuous subtree of $T$, since $c$ must be accessed in $\tau$ to reach $b$. 
	
	This leaves the case where accesses to $\tau$ are instead placed on the leftmost or rightmost column of $A$---that is, in addition to accesses that were already there. Consider an arbitrary access $a \in A$ and any access $b \in \tau$. If the accesses in $\tau$ are placed on a new column past the rightmost (leftmost) column of $A$, there is a monotonic path from $a$ to $b$ with accesses at every corner. If the accesses in $\tau$ are imposed on the rightmost (leftmost) column of $A$ instead, the same accesses still form a monotonic path, since this only causes the elimination of one right (left) move.
\end{proof}

We now prove correctness of the arboral mergesort algorithm.

\begin{lemma}\label{lem: arboralcorrectness}
The arboral mergesort algorithm is correct: that is, it returns an arborally satisfied set.
\end{lemma}

\begin{proof}
We will use the following inductive hypothesis on the arboral merge algorithm (Algorithm~\ref{algo: arboralmerge}) to show correctness: Algorithm~\ref{algo: arboralmerge} returns an arborally satisfied set which is double-congruent to the tree $T$ returned by the corresponding BST mergesort algorithm.

\textbf{Base Case.} When the set is just a single point, both arboral satisfaction and double congruence follow trivially.

\textbf{Inductive Step.} The inductive step is broken into several claims, and some new notation is called for. The two arborally satisfied sets being merged are $A$ and $B$, and by the inductive hypothesis are both arborally satisfied and double-congruent to trees $T_A$ and $T_B$, respectively. The additional accesses specified by the mergesort are added to three columns. Let $C_L$, $C_R$, and $C_M$ denote the set of points which the arboral merge adds
along the left, right, and middle columns respectively; see Figure~\ref{fig: arboralmerge} for an illustration.
It will also be useful to denote the subset of a $C_i$ ($i \in \{L,R,M\}$) consisting only of accesses to keys in $A$ or $B$. These subsets are denoted by $C_i(A)$ or $C_i(B)$.

The merge will break $A$ and $B$ into contiguous blocks that are interleaved in key order.   As assumed in the model, the block boundaries must be accessed by the merge, and therefore included in the $C_i$.  In general the $C_i$ will include other points as well.
The inductive step has to show both that the resulting set is arborally satisfied and is doubly congruent to the merged tree $T$.

\textbf{Arboral Satisfaction.}   
Points within $A$ or $B$ are satisfied by the inductive hypothesis.  Points both in $C_L, C_R, C_M$ are satisfied by the fact that they access precisely the same keys.  This leaves the two more interesting cases:
(1)  pairs of points one from the previously exiting points (in $A$ or $B$) and one from the new boundaries $C_L, C_R, C_M$,
and (2) pairs of points one from $A$ and one for $B$.
For the first case,  $A$ is double-congruent to tree $T_A$ (by the inductive hypothesis),
and $C_i(A)$ is a top tree of $A$ (by construction), so we can apply Observation~\ref{observation} for points in $A$ and $C_i(A)$.
Now since the boundaries of each block of $A$ must be in $C_i(A)$, we can get from a point
in $A$ to a point in $C_i(B)$ using a monotonic path by going to a boundary point in $C_i(A)$ and then up or down the column to the point in $C_i(B)$.
Symmetrically points in $B$ can get to points in $C_i(B)$ and  $C_i(A)$ by a monotonic path.
For case (2) consider any point $a \in A$ and point $b \in B$.
There is a monotonic path between $a$ and $b$ by composing the monotonic path between $a$ and a boundary point in
$p_a$ in $C_m(A)$, between $b$ and a boundary point $p_b$ in $C_m(B)$, and between $p_a$ and $p_b$, which are in the same
column; see Figure~\ref{fig: arboralmerge}. 

\textbf{Double congruency to $T$.}  
We will show that the set $AB$ returned by the arboral merge algorithm is right congruent to $T$.  Left congruency is true by symmetry.
The keys in $C_R$ (those accessed by the merge) correspond to a top tree $\tau$ of $T$.
Since all keys in the column $C_R$ have the same highest priority we can organize the root multi-node to match the structure
of $\tau$ making those nodes congruent.
Now consider the subtrees not in $\tau$.   They properly are lower in the tree and have equal or lower priority (only equal if they happen to be in the last row).  Furthermore the subtrees are separated by keys in $\tau$.  Each such subtree either comes
completely from $T_A$ or completely from $T_B$ and have the same structure as before the merge (they were not touched by the merge).   Furthermore the points from $T_A$ ($T_B$) only appear in $A$ ($B$).
This implies the relative priorities have not change for those points when merging into $AB$.   By induction, the trees
$T_A$ ($T_B$) were congruent to $A$ ($B$) before the merge so the subtrees were congruent and remain congruent after
the merge (neither the relative priorities nor tree structure have changed).  This implies the whole tree $T$ is 
congruent with $AB$.
\end{proof}

The proof of Theorem~\ref{thm: offlinealg} now follows easily.

\begin{proof}[Proof of Theorem~\ref{thm: offlinealg}]
The theorem follows once the cost of the arboral mergesort is shown to be $O(\mathcal{A}(\pi))$. Since the cost of $\mathcal{A}$ is dominated by the cost of each merge execution, and that cost is at most multiplied by six in each call to the arboral merge, the cost of the arboral mergesort is at most $6\mathcal{A}(\pi)$. When the arboral mergesort is transformed into an offline BST algorithm, the cost remains the same for a total cost of $O(\mathcal{A}(\pi))$.
\end{proof}

\subsection{Partition Sort}\label{sec: partitionsort}

We now consider a class of sorting algorithms motivated by quicksort, which we refer to as partition sorts.  As in the case of mergesort, we limit ourselves to working with binary trees.  However, in this case the trees are not ordered by key, but instead are ordered by input order.   The algorithm is like quicksort in that it picks a pivot, partitions the keys on the pivot and recurses.
Since we are interested in lower bounds (i.e. showing the cost of partition sort is at least as great as optimal BSTs), we can
assume an oracle picks the perfect pivot (e.g., the median).
As with mergesort, to achieve better than trivial $O(n \log n)$ bounds it is important that the partition need not visit the whole tree it is partitioning, but rather just some top tree. This allows for sending whole subsequences to the lesser or greater/equal side without
visiting all nodes.  More precisely here are the definitions of partition and partition sort.

\begin{definition}\label{def: btpartition}
  A \textbf{BT partition} takes a BT $T$ and for some top tree $\tau$ of $T$ returns two BTs $T_A, T_B$ with distinct keys such that for some top trees $\tau_a$ of $T_A$ and $\tau_b$ of $T_B$, $\tau = \tau_a \cup \tau_b$,  the other subtrees of $T$ appear in either $T_A$ or $T_B$ unchanged, and $\tau$ contains the block boundaries of the partitioned output.  Furthermore the preordering of the keys in $T_A$ or $T_B$ are a subsequence of the preordering in $T$ (i.e. left-to-right ordering).
 We assume both partitions are non-empty.
The \emph{number of accesses} used by the partition is $|\tau|$. 
\end{definition}

\begin{definition}\label{def: btpartitionsort} A \textbf{BT partition sort} on a BT tree $T$ (1) partitions $T$ into $T_a$ and $T_b$ such that for some key $k$ all keys in $T_a$ are less than $k$ and all keys in $T_b$ are greater or equal to $k$, (2) recurses on each partition, and (3) returns the left and right results appended.
The recursion terminates when the tree is of size one.
The number of accesses used by the partition sort is the sum of accesses across all partitions.  
\end{definition}

Our goal is to show the following theorem, which has the same form as the result for mergesort.

\begin{theorem}\label{thm: offlineps}
Let $\mathcal{A}$ be a BT partition sort algorithm that sorts permutation $\pi$ using $\mathcal{A}(\pi)$ accesses, and let  $\text{OPT}_{\text{BST}}(\pi)$ be the optimal cost of querying $\pi$ with a BST algorithm. Then $\text{OPT}_{\text{BST}}(\pi) \in O(\mathcal{A}(\pi))$. 
\end{theorem}

Our approach is to show a one-to-one correspondence between the tree-based merge and partition sorts.

\begin{lemma}\label{lem: btpartition}
For any BT partition sort algorithm  $\mathcal{A}$ that sorts permutation $\pi$ using $\mathcal{A}(\pi)$ accesses, 
there is a BST mergesort sort algorithm $\mathcal{B}$ that sorts permutation $\pi^{-1}$ using $\mathcal{B}(\pi^{-1})=\mathcal{A}(\pi)$ accesses, and vice versa.
\end{lemma}

\begin{proof}
  The idea is to consider running BST mergesort backwards, while reversing the role of time and key order.
  Consider undoing a merge---i.e. taking the merged tree and partitioning back into its two inputs.  Reversing the roles of time and keys, this is equivalent to a BT partition where keys are time order and the partitions is on the first time of the right partition.
In particular the  size of the top tree and therefore access cost is identical.   This continues to be true on the recursive calls.  
  In both cases the base case is of size one.   Hence the total access costs of the two algorithms are identical, one applied to the inverse permutation of the other.
\end{proof}

Since the size of arborally satisfied sets are invariant under rotation by 90 degrees, reversing the role does not 
affect the size of the set.    Since the proof of Theorem~\ref{thm: offlinealg} first showed how to map a BST mergesort to
an arborally satisfied set and this implied the same bound on a offline BST, this remains true if we rotate the input,
arborally satisfy it in the same way, and generate a BST.
Theorem~\ref{thm: offlineps} follows.
Taken together, Theorem~\ref{thm: offlinealg} and Theorem~\ref{thm: offlineps} show the statement in Theorem~\ref{thm: dualoffline}. Although the duality of mergesort and quicksort has been recognized before we are not aware of any formal 
correspondence such as the one given here.

\section{The \LIBcaps}
\label{sec: log-interleave}

The following two sections contain results that illustrate the utility of the approach shown in Theorem~\ref{thm: dualoffline}: namely, that results in the sorting cost model can directly translate to interesting results in the BST model. In this section we propose an information-theoretic bound on both the cost of accessing a sequence in the BST model and sorting a list in the comparison model. Theorem~\ref{thm: lib} shows that the \LIB{} is within a $\lg \lg n$ multiplicative factor of a known lower bound in the BST model. In the next section, we show that there exists a BST mergesort algorithm that sorts any permutation $\pi$ in $O(\lib(\pi))$ comparisons, and thus combined with Theorem~\ref{thm: offlinealg} shows the existence of an offline BST algorithm with the same costs in the BST model.

The \LIB{} can be thought of as an algorithmic perspective on Wilber's \IB. Let $P$ be the static tree with the keys of a permutation $\pi$ at the bottom. Consider sorting $\pi$ via mergesort: clearly, each non-leaf vertex of $P$ denotes a merge. The interleave bound charges unit cost for each switch between the right and left subtree during a merge. Another way of looking at this cost is that every continuous run of accesses to the left subtree incurs unit cost. Thus, a mergesort with a merge step that incurred unit cost for each consecutive run---as opposed to the standard mergestep which charges for the size of each run---would sort $\pi$ using $O(\ib(\pi))$ comparisons.

Lecomte and Weinstein~\cite{lecomte2020wilber} and Chalermsoook et al.~\cite{chalermsook2020pinning} independently show that the merge step described above does not exist. However, it is possible to charge the logarithm of the size of each consecutive run, as shown by Brown and Tarjan in~\cite{brown1980design}. This idea leads us to using such a merge step as an information-theoretic bound, which applies to sorting (sequentially and in parallel), and the BST model. The \LIB{} is formally defined below; note its similarity to the \IB.

\begin{figure}[t]
	\vspace{-.5em}
	\begin{center}
		\includegraphics[width=.5\columnwidth]{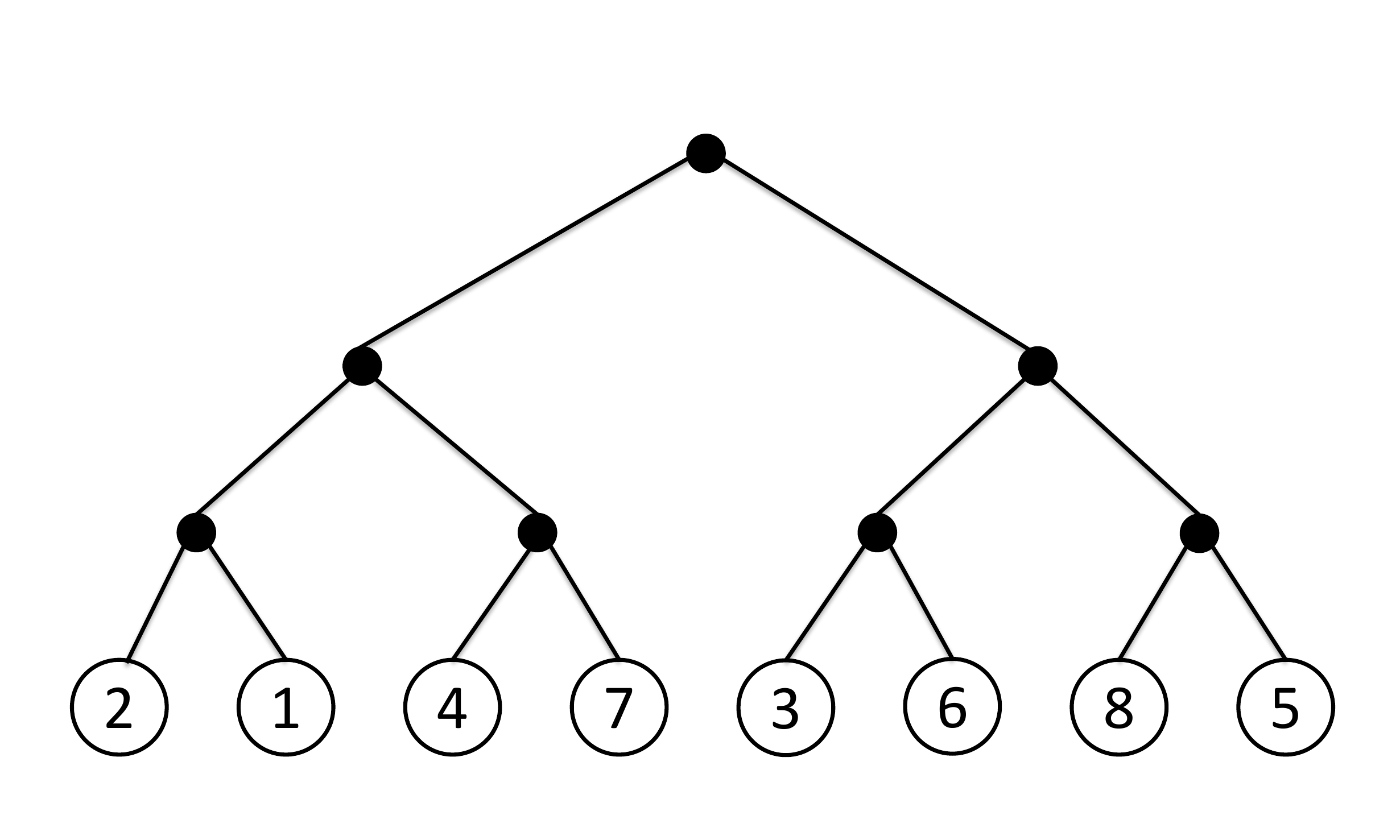}
	\end{center}\vspace{-1.5em}
	\caption{\small Consider accessing the keys 1-8 in order. For the vertex $v_1$ at the root of the tree shown here, the labeled access sequence is [L, L, L, R, L, R, R, R]. Since the access sequence switches between the left and right subtree three times, $\ib(v_1)=3$. Similarly, $\lib(v_1) = \lg 4 + \lg 2 + \lg 2 + \lg 4$ since the smallest possible decomposition of the labeled access sequence consists of [[L, L, L], [R], [L], [R, R, R]].}
	\label{fig: IB}
	\vspace{-1em}
\end{figure}

\begin{definition}\label{def: lib}
Given an access sequence $\pi$, fix a static binary tree $P$ with the keys of $\pi$ at the leaves. For each vertex $v_j$, query the descendants of $v_j$ in sorted order, then label each with R or L depending on whether it is in the left or right subtree of $v_j$. Let $S(v_j)$ represent the decomposition of this labeling into the smallest possible number of runs of consecutive accesses to L or R in $v_j$. Then $\lib(v_j) = \sum_{r_i \in S(v_j)} \lg (|r_i|+1)$ and $\lib(\pi) = \sum_{v \in P} \lib(v_i)$.
\end{definition}

See Figure~\ref{fig: IB} for an example calculation.

A natural question one might ask about a BST algorithm or an adaptive sorting algorithm is how far, in the worst case, is the cost of this algorithm from any known lower bounds? Or, in other words, how close is this algorithm to optimal? In this section, we will settle this question for the \LIB{} in the BST model, ending in the following theorem:

\thmlib

Our first step is to show that we cannot hope to do better than a $\lg \lg n$ separation; this is stated in the following lemma. This result is similar to the separation result in Theorem 2 of Lecomte and Weinstein~\cite{lecomte2020wilber}; furthermore, the result implies that an online BST algorithm using $\lib(\pi)$ accesses cannot be dynamically optimal. 

\begin{lemma}\label{lem: regandlib}
There exists a permutation $\pi$ such that $\lib(\pi) = \Theta(\lg \lg n \; \ib(\pi))$.
\end{lemma} 

\begin{proof}
	First we will need to define a particularly useful permutation. The \textit{bit-reversal permutation} $\pi_B$ on $n = 2^k$ keys is generated by taking a sorted list $[0, 1, \ldots, n]$, writing each key in binary, then reversing the bits of each key. For example, the bit-reversal permutation on 8 keys is $[0,4,2,6,1,5,3,7]$. Let $P$ be a static tree with keys of $\pi_B$ at the bottom: then querying its keys in sorted order will switch between the left and right subtree of any $v_j \in P$ on each query. This implies that $ \ib(\pi_B) = \Theta(n\lg n)$, thus showing that any BST algorithm will incur this cost when querying  $\pi_B$. 
	
	Consider the permutation $\pi$ obtained by splitting the sorted list into $n/\lg n$ segments of equal size, and then permuting those $n/\lg n$ segments according to the bit-reversal permutation $\pi_B$.
	
	The \IB{} of $\pi$ will be the same as for a list with $n/\lg n$ elements permuted according to the bit-reversal sequence---that is, $(n/\lg n) \lg (n/\lg n) = O(n)$. In $\pi$, every block is of size $\lg n$, so to calculate the \LIB{}, we multiply by $\lg \lg n$ on all but the bottom $\lg \lg n$ levels. Thus, the \LIB{} of $\pi$ is $\Theta(n \lg \lg n)$ while $\ib(\pi) = O(n)$.  
\end{proof}

Now we have shown there is no possibility of doing better than a $\lg \lg n$ separation, we show that this separation is tight. This starts with the following question: when are the \IB{} and the \LIB{} farthest apart? It follows from the convexity of the logarithm that for each vertex $v_j$ of a static tree, the \IB{} and the \LIB{} are farthest apart when $v_j$ experiences long runs of consecutive accesses to its subtrees (e.g. the list [L, L, L, R, R, R] has fairly different values for its \IB{} and \LIB{}, but the list [L, R, L, R, L, R] does not).

However, a completely sorted list $\pi_S$---translating to the longest run size possible for each vertex of $P$---has $\ib(\pi_S) = \lib(\pi_S) = \Theta(n)$. This suggests there must be some intermediate value of the block size that maximizes the difference between the two bounds. As the reader might have inferred from Lemma~\ref{lem: regandlib}, that size will turn out to be $\lg n$. The next lemma, whose proof follows directly from the convexity of the logarithm, formalizes this intuition.

\begin{lemma}\label{lem: equalsize}
	For a permutation $\pi$, let $v$ be a vertex of the corresponding static tree $P$ such that $\ib(v) = S$. Then $\lib(v)$ will differ from $\ib(v)$ by the greatest amount when each ``run'' of L or R in the labeled sequence is the same size.
\end{lemma}

Now that we have established that the \IB{} and the \LIB{} differ the most when all continuous runs are the same size, we move on to ask the following question: how large do the runs of the same size have to be to further maximize this difference? The next lemma shows this fact in the following way: in the inequality below, the expression $S \lg \left(\frac{n}{S}+1 \right)$ bounds the \LIB{} of any $\pi$ such that $\ib(\pi) = S$. 
The left-hand expression $c\lg (\lg n+1)\left(S + \frac{n}{\lg n} \right)$ will directly suffice to prove Theorem~\ref{thm: lib}. The proof of the lemma draws out the fact that the two expressions are closest to each other when the size of the continuous runs is $\lg n$.

\begin{lemma}\label{lem: ibcompare}
	For a permutation $\pi$, let $v$ be a vertex of the corresponding static tree $P$ such that $\ib(v) = S$. Furthermore, let the number of leaves below vertex $v$ be $n$. Then for some constant $c$,
	\begin{align*}
		c\lg (\lg n+1)\left(S + \frac{n}{\lg n} \right) \geq S \lg \left(\frac{n}{S}+1 \right).
	\end{align*}
\end{lemma}

\begin{proof}
	Assume for the sake of contradiction that 
	\begin{align*}
		c\lg(\lg n + 1)S + c\lg(\lg n+1)\frac{n}{\lg n} < S \lg \left(\frac{n}{S}+1 \right).
	\end{align*}
	This would imply that each added term is smaller than $S \lg \left(\frac{n}{S}+1 \right)$. Begin by examining the case where the first term is smaller than the term on the right:
	\begin{align*}
		\lg(\lg n +1)S < S\lg\left( \frac{n}{s}+1 \right) \\
		\implies \lg (\lg n+1) < \lg \left( \frac{n}{S}+1 \right) \\
		\implies \lg n < \frac{n}{S} \\
		\implies S < \frac{n}{\lg n}.
	\end{align*} 
	This shows that when $S \geq \frac{n}{\lg n}$, we reach a contradiction and our claim holds. Now, when $S < \frac{n}{\lg n}$, the second term in the sum dominates. When the second term dominates, the expression reads
	\begin{align*}
		c \lg(\lg n+1) \cdot \frac{n}{\lg n} \geq \frac{n}{\lg n}\lg (\lg n+1)
	\end{align*}
	which is self-evidently true for all $c \geq 1$. 
\end{proof}

Now, these two lemmas are put together to prove Theorem~\ref{thm: lib}. 

\begin{proof}[Proof of Theorem~\ref{thm: lib}]
	Let $P$ be the static tree corresponding to $\pi$, and for each vertex $v_i$ of $P$, let $S_i = \ib(v_i)$. By Lemma~\ref{lem: equalsize}, we can assume that if the number of leaves below $v_i$ is $n_i$, then $\lib(v_i) = S_i\lg \left(\frac{n_i}{S}+1 \right)$. Then $\ib(\pi) = \sum_{i=1}^{n-1} S_i$. Next, we can use the upper bound on $\lib(v_i)$ from Lemma~\ref{lem: ibcompare} to upper bound $\lib(\pi)$:
	\begin{align*}
		\lib(\pi) \leq \sum_{i=1}^{n-1} c \lg \lg n \left(S + \frac{n_i}{\lg n} \right) \\
		= c \lg \lg n \left(\ib(\pi) + \frac{1}{\lg n}\sum_{i=1}^{\lg n} 2^i \frac{n}{2^i} \right) \\
		= c \lg \lg n \left(\ib(\pi) + n \right) = O(\ib(\pi) \lg \lg n).
	\end{align*}
        \vspace{-.2in}
      \end{proof}

\section{\APMcaps}\label{apdx: apm}
In this section we present a parallel BST mergesort which sorts a permutation $\pi$ using $O(\lib(\pi))$ accesses, and the same amount of work. We refer to the algorithm as an \textit{adaptive parallel mergesort}. 

First we introduce a few basic terms and the data structure used in our mergesort. Given a BST $T$ and a key $k$, a \textbf{split} refers to returning two BSTs, one containing all keys from $T$ which are greater than $k$, and one containing all keys which are less than $k$. Given two BSTs $T_1, T_2$ such that any key in $T_1$ is greater than every key in $T_2$, \textbf{join} returns a single BST $T$ containing the union of the keys in $T_1$ and $T_2$. As previously stated, we assume keys are unique.

The tree used in our mergesort algorithm is a modified red-black tree described by Tarjan and Van Wyck in~\cite{tarjan1988triangulating}, which they call a \textit{heterogeneous finger search tree}. These trees have the useful property that a key $d$ in a heterogenous finger search tree with $n$ elements can be accessed in time $O(\log(\min(d, n-d)+1))$. This property allowed Tarjan and Van Wyck to devise fast split and join algorithms for heterogenous finger search trees; split runs in amortized time $O(\lg(\min(|T_1|, |T_2|)+1))$---that is, the logarithm of the size of the smaller tree returned. Join similarly is bounded by amortized time $O(\lg(\min(|T_1|, |T_2|)+1))$---in this case, the size of the smaller of the two trees being joined together. The worst-case complexity of split and join is $O(\log \max{|T_1|, |T_2|})$. As presented in~\cite{tarjan1988triangulating}, the heterogeneous finger search tree is not strictly a BST as it uses more than one pointer; however, work by~\cite{chalermsook2018multi} shows how it can be converted into using a single pointer with an additional constant factor loss.

\textbf{Sequential Mergesort.} As a warmup, we present a sequential merge algorithm which, when used in as the merge step of a mergesort algorithm, sorts a permutation $\pi$ in $O(\lib(\pi))$ time. Refer to Algorithm~\ref{algo: seqmerge} for the algorithm. We note that McIlroy also proposed a sequential comparison-based sorting algorithm that uses $O(\lib(\pi))$ comparisons~\cite{mcilroy1993sorting}, but his algorithm uses a merge step which requires linear time in the size of the lists being merged. 

\begin{algorithm2e}
	\caption{merge($T_1, T_2$).}
	\label{algo: seqmerge}\small
	\SetAlgoLined
	\KwIn{Two BSTs $T_1, T_2$ with disjoint keys.}
	\KwOut{A BST containing the union of the keys of $T_1$ and $T_2$.}
	$T = \emptyset$ \;
	\While {true} {
		\lIf{$T_1 =$ leaf}{
			\KwRet{join ($T, T_2$)}
		}
		\lElseIf{$T_2 = $leaf}{
			\KwRet{join ($T, T_1$)}
		}
		\Else{
			$k_1 = \min(T_1)$;
			$k_2 = \min(T_2)$\;
			\uIf{$k_1 > k_2$}{
				$t, T_2 = $ split($T_2, k_1$) ;
				$T = $ join($T, t$)\;
			}
			\Else{
				$t, T_1 = $ split($T_1$, $k_2$) ; 
				$T = $ join($T, t$) \;
			}
		}
	}
	\vspace{0.5em}
\end{algorithm2e}

\begin{lemma}\label{lem: seqmerge}
	Algorithm~\ref{algo: seqmerge} is correct and, if it breaks $T_1$ into subtrees $t_1, \ldots, t_j$ and $T_2$ into subtrees $t_{j+1}, \ldots, t_k$, then it runs in time 
	\begin{align*}
		O \left( \sum_{i=1}^k \lg (|t_i|+1) \right).
	\end{align*}
\end{lemma}

\begin{proof}
	The proof of the time bound is immediate, since a heterogeneous finger search tree has split and join cost $O(\lg(|t|+1))$ for every subtree that it breaks off from $T_1$ and $T_2$.
	
	Correctness will be shown via induction. The inductive hypothesis is that after each split and join of a new subtree, $T$ remains a valid BST, and each key of $T$ is smaller than every key in $T_1 \cup T_2$. 
	
	The base case is when $T$ is empty. Assume without loss of generality that $k_1 > k_2$. Then $\text{split}(T_2, k_1)$ returns $T_2$ and $t$, where $t$ contains only keys which are smaller than all keys in $T_2$. They are also smaller than all keys in $T_1$, since $k_1$ was the minimum key in $T_1$ and $t$ contains only keys less than $k_1$. The operation $\text{join}(T, t)$ successfully returns the BST $t$ since $T = \emptyset$. 
	
	Now, assume that after any number of split/join actions that $T$ is a valid BST and contains only keys that are less than those in $T_1 \cup T_2$. Similarly to the proof of the base case, assuming wlog that $k_1 >k_2$, we know that the split of $T_2$ returns a tree $t$ such that each key in $t$ is larger than any key in $T$, and such that each key in $t$ is smaller than any key in $T_1 \cup T_2$. Furthermore, since $t$ contains only keys larger than those in $T$, $\text{join}(T, t)$ returns a valid BST.  
\end{proof}

Next, we show that the mergesort algorithm which uses Algorithm~\ref{algo: seqmerge} as its merge routine is optimal for the \LIB.

\begin{lemma}\label{lem: seqmergesort}
	When Algorithm~\ref{algo: seqmerge} is used as the merge step in a mergesort algorithm, the mergesort will sort a sequence $\pi$ in time $O(\lib(\pi))$.
\end{lemma}

\begin{proof}
	When merging two trees $T_1, T_2$, consider any permutation $\pi$ such that the keys of $T_1$ correspond to the keys in the first half of $\pi$, and the keys of $T_2$ correspond to the keys in the second half of $\pi$. 
	
	Consider the static tree $P_\pi$ used to calculate the \LIB{} of $\pi$. The root vertex $v$ then contains all the keys in $T_1$ in its left subtree, and all the keys in $T_2$ in its right subtree. Thus
	\begin{align*}
		\lib(v) = \sum_{i=1}^k \lg (|t_i|+1)	
	\end{align*}
	since the $t_i$'s correspond to switching between $T_1$ and $T_2$ when querying the keys in sorted order. 
	
	Finally, observe that each non-leaf vertex of $P_\pi$ corresponds to a merge that a mergesort algorithm would carry out, and the result follows.
\end{proof}

The natural parallel equivalent of the merge presented in Algorithm~\ref{algo: seqmerge} is as follows: starting with two trees, split each tree using the other tree's root; then, recurse in parallel to merge the two left halves and the two right halves, respectively, joining the two at the end. This idea was presented by Blelloch et al.~\cite{blelloch2016just}, and is shown here in Algorithm~\ref{algo: union}. This algorithm, however, does not meet the \LIB{} even if we use heterogeneous finger search trees for the split and join. We therefore modify the algorithm as is shown in Algorithm~\ref{algo: merge} and illustrated in Figure~\ref{fig: merge}, which follows the same idea with some small modifications.   In addition to splitting the second tree $T_2$ into $L_2$ and $R_2$ based on the root of the first tree ($T_1$), it then splits $T_1$ by the maximum value of $L_2$ and the minimum value of $R_2$ to effectively break $T_1$ into three parts. The middle part need not be split recursively since it falls between two elements of $T_2$. This avoids redundant splits.

\begin{minipage}{.46\textwidth}
	\begin{algorithm2e}[H]
		\caption{union($T_1, T_2$).\protect\\
			Blelloch et al.'s union algorithm. Here, the function \textit{expose} refers to returning the root and its right and left subtrees.}
		\label{algo: union}\small
		\SetKwBlock{ParDo}{do in parallel}{end}
		\SetAlgoLined
		\KwIn{Two BSTs $T_1, T_2$ with disjoint keys.}
		\KwOut{A BST containing the union of the keys of $T_1$ and $T_2$}
		\lIf{$T_1 = $ Leaf}{
			\KwRet{$T_2$}
		}
		\lElseIf{$T_2 = \text{Leaf}$}{
			\KwRet{$T_1$}
		}
		\Else{
			$L_1, k_2, R_1 = $ expose($T_1$)\;
			$L_2, R_2 = $ split($T_2, k$) \;
			\ParDo{
				$T_L = $ union($L_1, L_2$)\; 
				$T_R = $ union($R_1, R_2$) \;
			}
		}	
		\KwRet{join($T_L, k_2, T_R$)}\;
		\vspace{0.5em}
	\end{algorithm2e}
\end{minipage} \hfill
\begin{minipage}{.46\textwidth}
	\begin{algorithm2e}[H]
		\caption{mergeHT($T_1,T_2)$\protect\\
			Pseudocode for the merge step of our mergesort.}
		\label{algo: merge}\small
		\SetKwBlock{ParDo}{do in parallel}{end}
		\SetAlgoLined
		\KwIn{Two BSTs $T_1, T_2$}
		\KwOut{A BST containing the union of the keys of $T_1$ and $T_2$}
		\lIf{$T_1 = $ Leaf}{
			\KwRet{$T_2$}
		}
		\lElseIf{$T_2 = \text{Leaf}$}{
			\KwRet{$T_1$}
		}
		\Else{
			$k$ = root($T_1$)\;\label{line:root}
			$L_2, R_2 = $ split($T_2, k$) \;
			$k_1 = $ max($L_2$) ; 
			$k_2 = $ min($R_2$) \; 
			$L_1, I= $ split($T_1, k_1$) \; 
			$M, R_1 = $ split($I$, $k_2$) \; 
			\ParDo{
				$T_L = $ merge($L_1, L_2$) \; 
				$T_R = $ merge($R_1, R_2$) \;
			}
		}
		
		\KwRet{join(join($T_L, M$), $T_R$)}\;
		\vspace{0.5em}
	\end{algorithm2e}
\end{minipage}

Unlike the sequential algorithm, it is not immediate that even our modified algorithm's work is bounded by the \LIB, since it performs a different set of splits and joins than the sequential version. We will show that this different sequence of splits and joins also performs within the \LIB, culminating in the following theorem:

\thmsort

\begin{figure*}[t!]
	\vspace{-.5em}
	\begin{center}
		\includegraphics[width=.30\columnwidth]{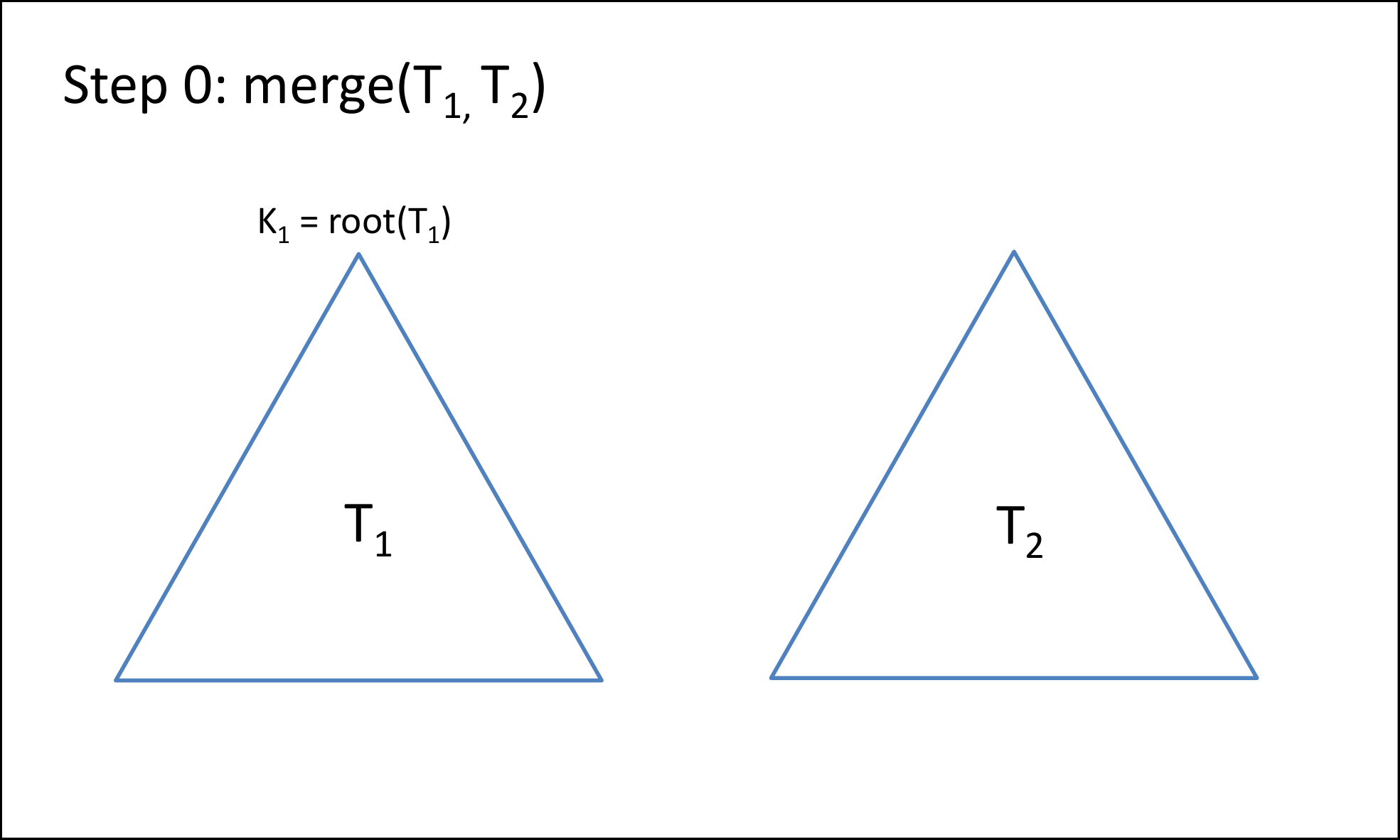}
		\includegraphics[width=.30\columnwidth]{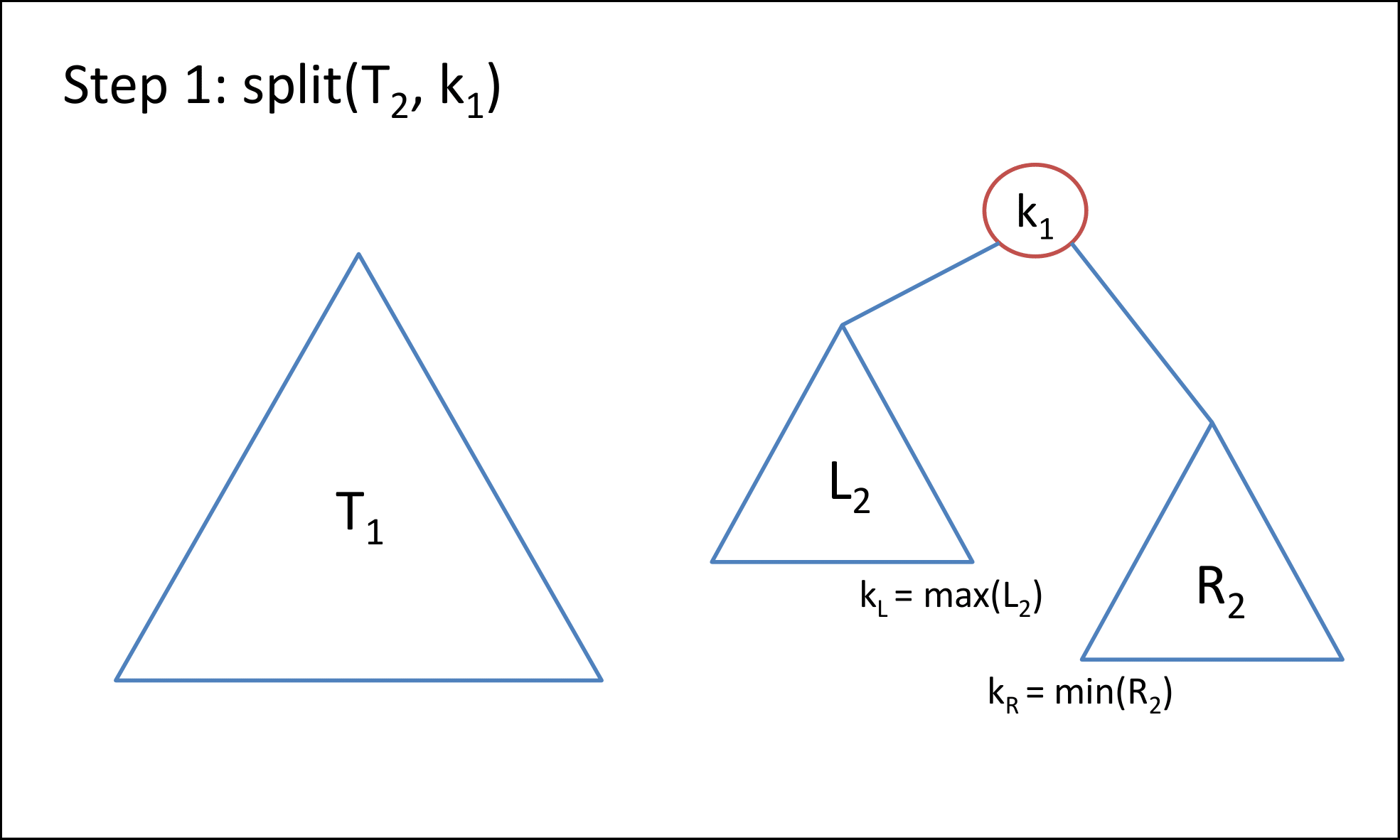}
		\includegraphics[width=.30\columnwidth]{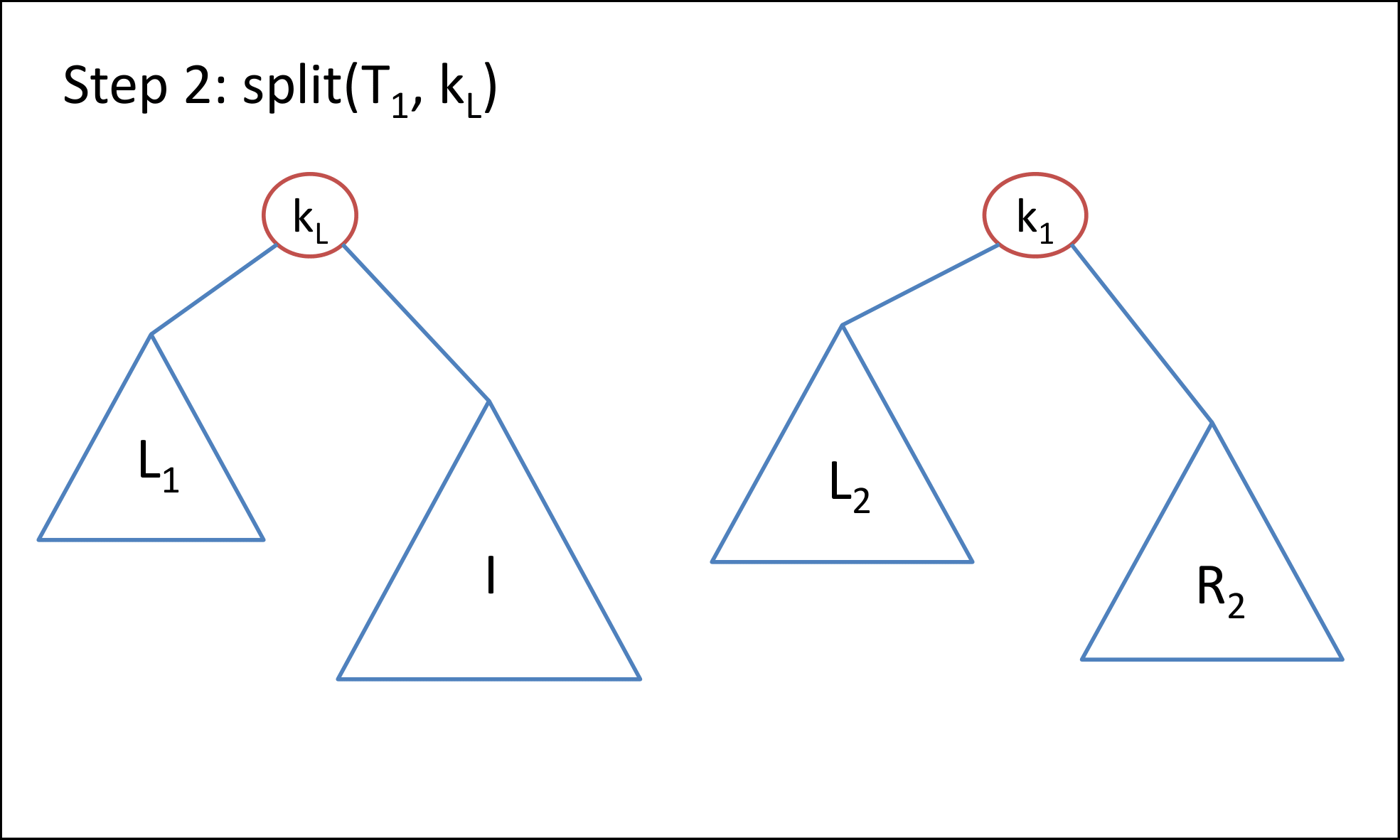}
		\includegraphics[width=.30\columnwidth]{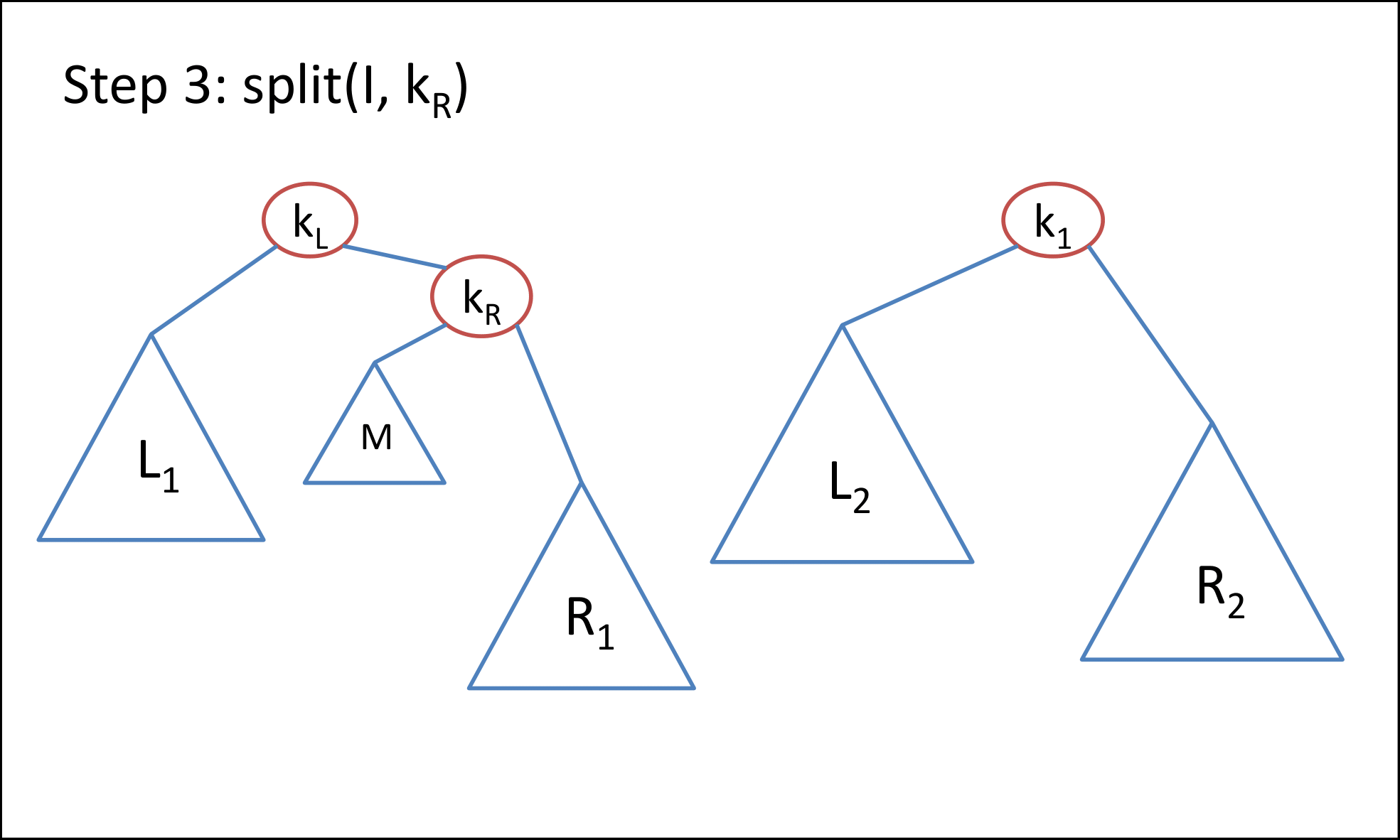}
		\includegraphics[width=.30\columnwidth]{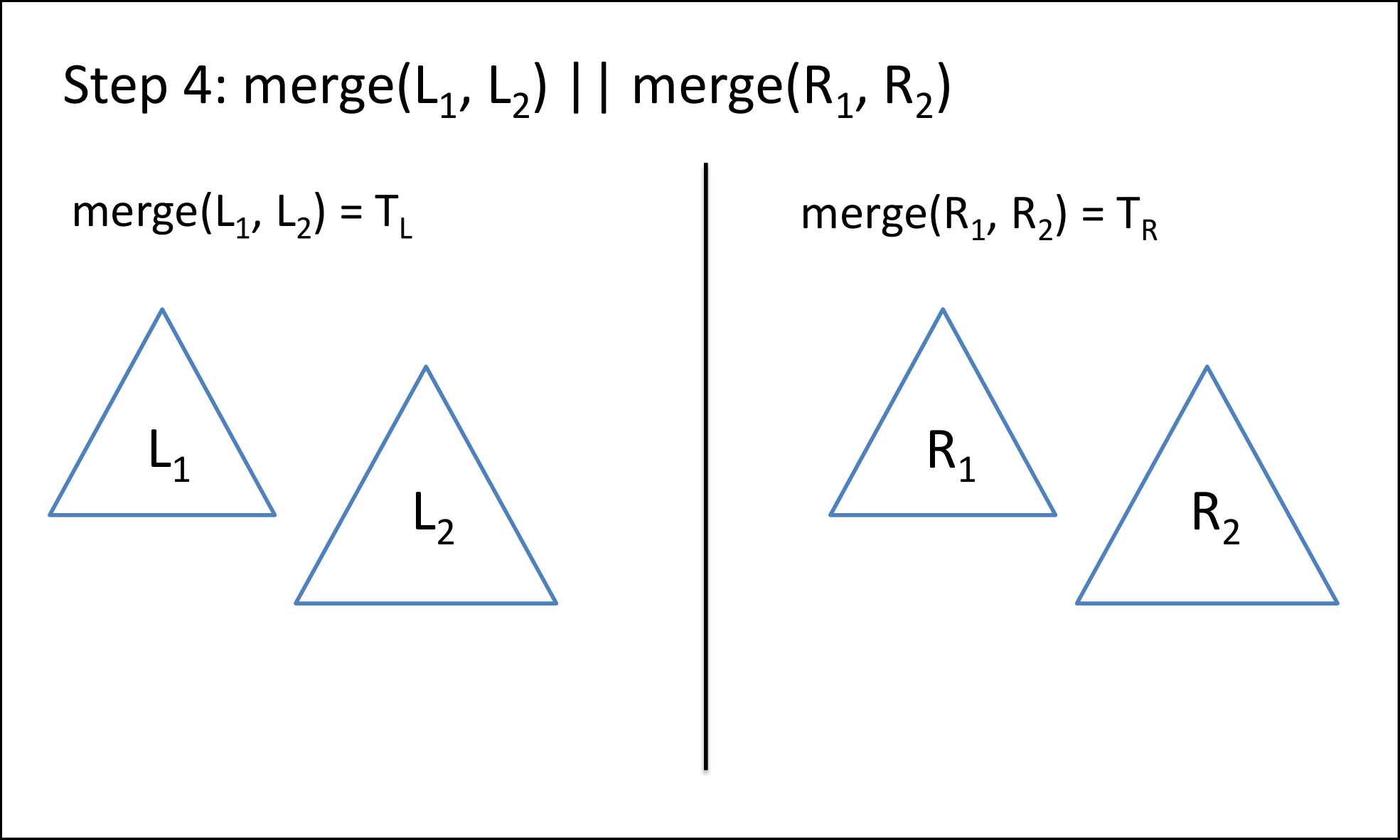}
		\includegraphics[width=.30\columnwidth]{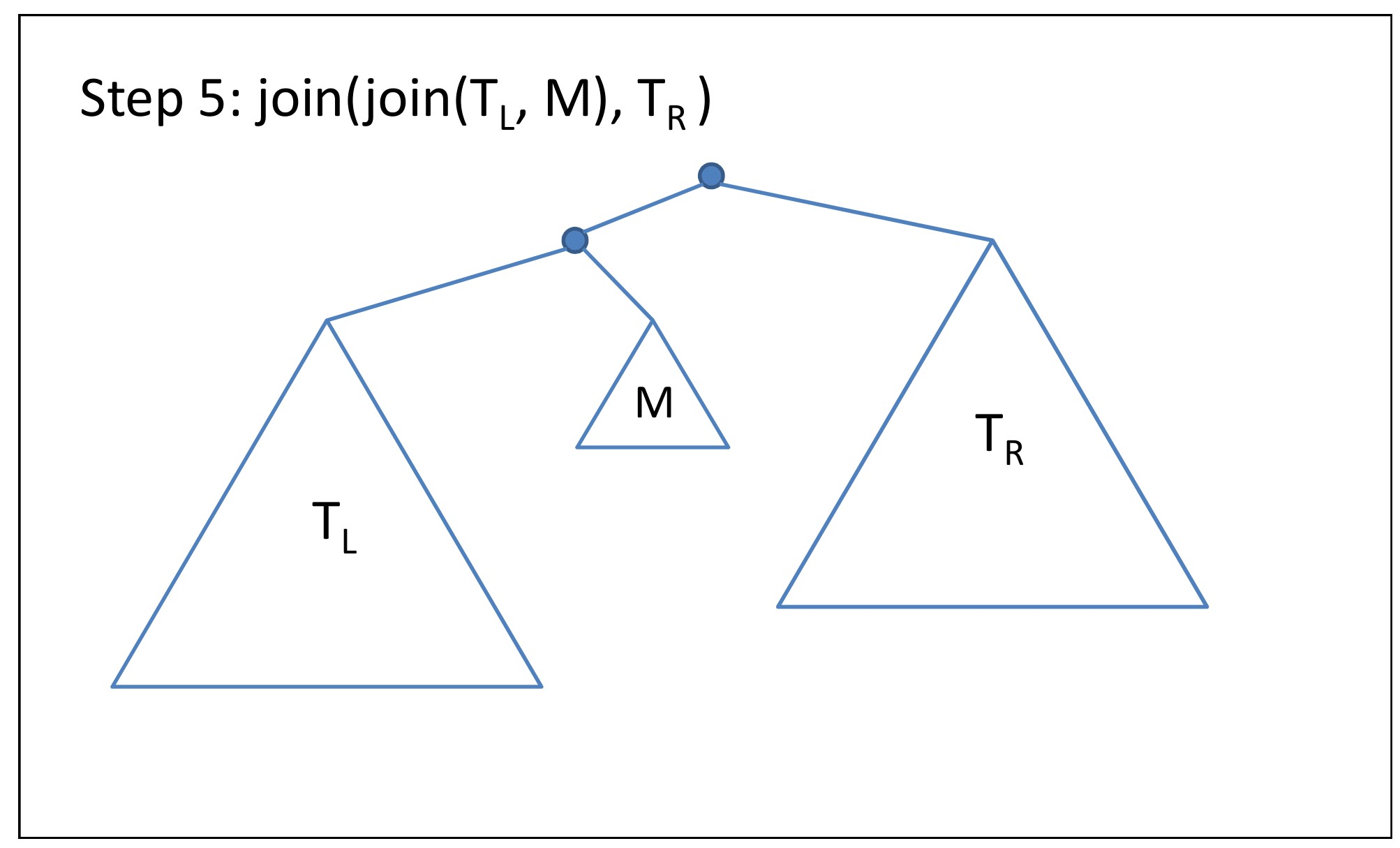}
	\end{center}
	\caption{\small One round of our recursive merge algorithm. The nodes
		shown in red are the nodes used to split a tree; the small blue
		nodes denote merges.}
	\label{fig: merge}
\end{figure*}

First, we introduce a useful lemma which shows that the work performed to split a heterogeneous finger search tree into $k$ smaller trees is independent of the order in which the splits are carried out.

\begin{lemma}\label{lem: mergesort}
	Let $T$ be a heterogeneous finger search tree and consider splitting $T$ into $k$ smaller trees $t_1, \ldots, t_k$. Then the work required to split $T$ is $O\left(\sum_{i=1}^k \lg(|t_i|+1) \right)$ regardless of the order in which the splits are performed.
\end{lemma}

\begin{proof}
	The proof proceeds by induction on the size. Let $|T|=n$, and make an inductive hypothesis that the total split work of $T$ is
	\begin{align*}
		\left( \sum_{i=1}^k \lg (|t_k|+2) \right) - (\lg n +2).
	\end{align*}
	Begin by verifying the base case: splitting a BST of size 1 should have a cost of zero:
	\begin{align*}
		\lg(1)+2 - (\lg(1)+2) = 0.
	\end{align*}
	Now, for the inductive step, assume we are splitting $T$ into two parts of size $\alpha n$ and $(1-\alpha)n$ for $\alpha \in [0,1]$, and that the first part consists of the smaller trees $t_1, \ldots, t_j$ while the second part consists of the trees $t_{j+1}, \ldots, t_k$. Then our inductive step requires that:
	\begin{align*}
		\left(\sum_{i=1}^j \lg |t_i| - \lg(\alpha n)-2 \right) + \left(\sum_{i=j+1}^k \lg |t_i| - \lg((1-\alpha))n-2 \right) \\ +  \lg(\min(\alpha n, (1-\alpha)n))+1 \leq \sum_{i=1}^k \lg |t_i| - \lg n -2 
	\end{align*}
	Assume without loss of generality that the size of the first subtree is smaller than the second, i.e. $\alpha \in [0, \frac{1}{2}]$. Then canceling on both sides and substituting in the minimum value gives:
	\begin{align*}
		-\lg (\alpha n) - \lg ((1-\alpha)n)-4+\lg \alpha n+1 \leq -\lg n -2 \\
		\lg ((1-\alpha)n)+1 \geq \lg n.
	\end{align*}
	Finally, since $1-\alpha \in [\frac{1}{2}, 1]$, in the case where the right-hand expression is smallest and $\alpha=\frac{1}{2}$, $\lg (\frac{1}{2}n)+1 = \lg n$, so the inequality is both true and tight in the worst case. 
	
\end{proof}

This lemma allows us to show the correctness and the work bounds for Algorithm~\ref{algo: merge}.

\begin{lemma}\label{lem: mergecorrect}
	Algorithm~\ref{algo: merge} is correct, and if merging $T_1$ and $T_2$ requires $T_1$ and $T_2$ to be split into $k$ trees $t_1, \ldots, t_k$ and $j$ smaller trees $t_1, \ldots, t_j$ respectively, then Algorithm~\ref{algo: merge} performs $O(\sum_{i=1}^k \lg|t_i+1| + \sum_{i=1}^j \lg |t_i+1|)$ work. 
\end{lemma}

\begin{proof}
	Correctness follows from the fact that by the assumptions on the input, the interval $I$ contains all keys between those of $L_1, L_2$ and $R_1, R_2$, and from the correctness of Algorithm~\ref{algo: union}. 
	
	The work bound comes from Lemma~\ref{lem: mergesort}, which shows that it does not matter in what order the splits are performed, and therefore they are bounded by the sequential work. Since each join is the same cost as a split, the proof of Lemma~\ref{lem: mergesort} also suffices to show that the join work is bounded by the split work, so the work bound follows.
\end{proof}

Note that the position of the ``root'' in the tree $T_1$ as selected in line~\ref{line:root} does not matter for the purpose of the work bounds.   Picking the first element, for example, would lead to an algorithm that is very similar to the sequential algorithm.     However for the purpose of parallelism it is important that the root in near the middle---i.e., cuts off at least a constant fraction of $T_1$'s keys from each side.   This could be achieved by finding the median element, and the cost would be within bounds.

Our last lemma concerns the span of the parallel mergesort algorithm.

\begin{lemma}\label{lem: mergespan}
	The span of the parallel adaptive mergesort is $O(\lg^3 n)$. 
\end{lemma}

\begin{proof}
	Our merge has the same span as Blelloch et al.'s union algorithm since split and join have the same worst-case guarantees; they proved in~\cite{blelloch2016just} that the span of the merge is $O(\lg^2 n)$ when merging two trees of size $n$. Since parallel mergesort has span $O(\lg n)$, the total span is $O(\lg^3 n)$. 
	We only require binary forking since the only parallel calls are the two recursive calls inside merge, and the two recursive calls 
	inside mergesort.
\end{proof}

Now, we can put all these results together to prove our theorem.

\begin{proof}[Proof of Theorem~\ref{thm: mergesort}]
	Lemma~\ref{lem: mergesort} shows that the merge step in our parallel mergesort performs the same amount of work as the mergestep in the sequential mergesort from Lemma~\ref{lem: seqmergesort}; Lemma~\ref{lem: mergecorrect} shows correctness. Lemma~\ref{lem: mergespan} shows the span.
\end{proof}

A natural corollary of Theorem~\ref{thm: mergesort} and Theorem~\ref{thm: offlinealg} now follows. Notably, this corollary shows that the \LIB{} is in fact an upper bound in the BST model.

\corlib

\begin{proof}[Proof of Corollary~\ref{cor: offlinelib}]
	Theorem~\ref{thm: mergesort} gives a BST mergesort that has the desired bounds; thus Theorem~\ref{thm: offlinealg} implies the existence of an offline BST algorithm.
\end{proof}



\bibliography{../bibliography/strings,../bibliography/main}

\appendix

\end{document}